\documentclass[reqno,11pt]{amsart}
\usepackage{amsmath, latexsym, amsfonts, amssymb, amsthm, amscd}
\usepackage{graphics,epsf,psfrag}

\numberwithin{equation}{section}

\newtheorem{theorem}{Theorem}[section]
\newtheorem{lemma}[theorem]{Lemma}
\newtheorem{proposition}[theorem]{Proposition}
\newtheorem{cor}[theorem]{Corollary}
\newtheorem{rem}[theorem]{Remark}

\newcommand{\ind}{\mathbf{1}}

\newcommand{\R}{\mathbb{R}}
\newcommand{\Z}{\mathbb{Z}}
\newcommand{\N}{\mathbb{N}}
\renewcommand{\tilde}{\widetilde}
\renewcommand{\hat}{\widehat}

\newcommand{\cG}{{\ensuremath{\mathcal G}} }
\newcommand{\cN}{{\ensuremath{\mathcal N}} }

\newcommand{\cI}{{\ensuremath{\mathcal I}} }

\newcommand{\bP}{{\ensuremath{\mathbf P}} }
\newcommand{\bE}{{\ensuremath{\mathbf E}} }

\newcommand{\bL}{{\ensuremath{\mathbf L}} }

\DeclareMathSymbol{\leqslant}{\mathalpha}{AMSa}{"36} % nicer `smaller or equal'
\DeclareMathSymbol{\geqslant}{\mathalpha}{AMSa}{"3E} % nicer `larger or equal'
\DeclareMathSymbol{\eset}{\mathalpha}{AMSb}{"3F}     % nicer `emptyset'
\newcommand{\dd}{\,\text{\rm d}}             % a straight d for differentials
\newcommand{\sumtwo}[2]{\sum_{\substack{#1 \\ #2}}} % sum with 2 lines
\newcommand{\bbE}{{\ensuremath{\mathbb E}} }

\newcommand{\bbP}{{\ensuremath{\mathbb P}} }

\newcommand{\bbZ}{{\ensuremath{\mathbb Z}} }

\newcommand{\ga}{\alpha}
\newcommand{\gb}{\beta}
\newcommand{\gd}{\delta}
\newcommand{\gep}{\varepsilon}       % \ge already exists...

\newcommand{\gD}{\Delta}

\newcommand{\go}{\omega}

\newcommand{\gl}{\lambda}

\makeatletter
\def\captionfont@{\footnotesize}
\def\captionheadfont@{\scshape}

\long\def\@makecaption#1#2{%
  \vspace{2mm}
  \setbox\@tempboxa\vbox{\color@setgroup
    \advance\hsize-6pc\noindent
    \captionfont@\captionheadfont@#1\@xp\@ifnotempty\@xp
        {\@cdr#2\@nil}{.\captionfont@\upshape\enspace#2}%
    \unskip\kern-6pc\par
    \global\setbox\@ne\lastbox\color@endgroup}%
  \ifhbox\@ne % the normal case
    \setbox\@ne\hbox{\unhbox\@ne\unskip\unskip\unpenalty\unkern}%
  \fi
  \ifdim\wd\@tempboxa=\z@ % this means caption will fit on one line
    \setbox\@ne\hbox to\columnwidth{\hss\kern-6pc\box\@ne\hss}%
  \else % tempboxa contained more than one line
    \setbox\@ne\vbox{\unvbox\@tempboxa\parskip\z@skip
        \noindent\unhbox\@ne\advance\hsize-6pc\par}%
\fi
  \ifnum\@tempcnta<64 % if the float IS a figure...
    \addvspace\abovecaptionskip
    \moveright 3pc\box\@ne
  \else % if the float IS NOT a figure...
    \moveright 3pc\box\@ne
    \nobreak
    \vskip\belowcaptionskip
  \fi
\relax
}
\makeatother
\def\writefig#1 #2 #3 {\rlap{\kern #1 truecm
\raise #2 truecm \hbox{#3}}}

\newcommand{\tf}{\textsc{f}}
\newcommand{\M}{\textsc{M}}

\newcommand{\ui}{\underline{i}}
\newcommand{\uj}{\underline{j}}
\newcommand{\um}{\underline{m}}
\newcommand{\sort}{\mathtt{s}}
\newcommand{\mbeta}{{\mathtt m}_\gb}
\newcommand{\const}{{\mathtt c}}
\newcommand{\consta}{{\mathtt a}}
\newcommand{\constb}{{{\mathtt b}}}

\begin{document}
\title[Disorder relevance at marginality]{Disorder relevance at marginality \\
and critical point shift}

\author{Giambattista Giacomin}
\address{
  Universit{\'e} Paris Diderot (Paris 7) and Laboratoire de Probabilit{\'e}s et Mod\`eles Al\'eatoires (CNRS),
U.F.R.                Math\'ematiques, Case 7012 (site Chevaleret)
             75205 Paris cedex 13, France
}
\email{giacomin\@@math.jussieu.fr}
\author{Hubert Lacoin}
\address{
  Universit{\'e} Paris Diderot (Paris 7) and Laboratoire de Probabilit{\'e}s et Mod\`eles Al\'eatoires (CNRS),
U.F.R.                Math\'ematiques, Case 7012 (site Chevaleret)
             75205 Paris cedex 13, France
}
\email{lacoin\@@math.jussieu.fr}
\author{Fabio Lucio Toninelli}
\address{CNRS and 
Laboratoire de Physique, ENS Lyon, 46 All\'ee d'Italie, 
69364 Lyon, France
}
\email{fabio-lucio.toninelli@ens-lyon.fr}
\date{\today}

\begin{abstract}
  Recently the renormalization group predictions on the effect of
  disorder on pinning models have been put on mathematical grounds.
  The picture is particularly complete if the disorder is {\sl
    relevant} or {\sl irrelevant} in the Harris criterion sense: the
  question addressed is whether quenched disorder leads to a critical
  behavior which is different from the one observed in the pure, {\sl
    i.e.} annealed, system. The Harris criterion prediction is based
  on the sign of the specific heat exponent of the pure system, but it
  yields no prediction in the case of vanishing exponent. This case is
  called {\sl marginal}, and the physical literature is divided on
  what one should observe for marginal disorder, notably there is no
  agreement on whether a small amount of disorder leads or not to a
  difference between the critical point of the quenched system and the
  one for the pure system.  In \cite{cf:GLTmarg} we have proven that
  the two critical points differ at marginality of at least
  $\exp(-c/\gb^4)$, where $c>0$ and $\gb^2$ is the disorder variance,
  for $\gb \in (0,1)$ and Gaussian IID disorder.  The purpose of this
  paper is to improve such a result: we establish in particular that
  the $\exp(-c/\gb^4)$ lower bound on the shift can be replaced by
  $\exp(-c(b)/\gb^b)$, $c(b)>0$ for  $b>2$ ($b=2$ is the known upper bound and
  it is the result claimed in \cite{cf:DHV}), and we deal with very
  general distribution of the IID disorder variables.  The proof
  relies on coarse graining estimates and on a fractional
  moment--change of measure argument based on multi-body potential
  modifications of the law of the disorder. 
 % In addition, we prove a monotonicity result on the phase diagram of pinning model with no assumption on the renewal--process.
  \\
  \\
  2000 \textit{Mathematics Subject Classification: 82B44, 60K35, 82B27, 60K37
  }
  \\
  \\
  \textit{Keywords: Disordered Pinning Models, Harris Criterion,
    Marginal Disorder, Many-body interactions }
\end{abstract}

\maketitle

\section{introduction}

\subsection{Relevant, irrelevant and marginal disorder}
\label{sec:RIM}
The renormalization group approach to disordered statistical mechanics
systems introduces a very interesting viewpoint on the role of
disorder and on whether or not the critical behavior of a quenched
system coincides with the critical behavior of the corresponding
%%F
%annealed system, also called
{\sl pure} system.  The Harris criterion \cite{cf:Harris} is based on such an
approach and it may be summarized in the following way: if the
specific heat exponent of the pure system is negative, then a {\sl
  small} amount of disorder does not modify the critical properties of
the pure system ({\sl irrelevant disorder regime}), but if the
specific heat exponent of the pure system is positive then even an
arbitrarily small amount of disorder may lead to a quenched critical
behavior different from the critical behavior of the pure system.
 
A class of disordered models on which such ideas have been applied by
several authors is the one of pinning models (see {\sl e.g.}
\cite{cf:FLNO,cf:DHV} and the extensive bibliography in
\cite{cf:Book,cf:GLTmarg}).  The reason is in part due to the
remarkable fact that pure pinning models are %a class of
exactly solvable models for which, by tuning a parameter, one can
explore all possible values of the specific heat exponent
\cite{cf:Fisher}. As a matter of fact, the validity of Harris
criterion for pinning models in the physical literature finds a rather
general agreement. Moreover, for the pinning models the
renormalization group approach goes beyond the critical properties and
yields a prediction also on the location of the critical point.
 
Recently, the Harris criterion predictions for pinning models have
been put on firm grounds in a series of papers \cite{cf:Ken,cf:T_cmp,cf:DGLT,cf:AZ2} and some of these
rigorous results go even beyond the predictions. Notably in
\cite{cf:GT_cmp} it has been shown that disorder has a  smoothing effect in 
this class of models (a fact that is not a consequence of the Harris
criterion and that does not find unanimous agreement in the physical literature).
 
However, a substantial amount of the literature on disordered pinning
and Harris criterion revolves around a specific issue: what happens if
the specific heat exponent is zero ({\sl i.e.} at {\sl marginality})?
This is really a controversial issue in the physical literature,
started by the disagreement in the conclusions of \cite{cf:FLNO} and
\cite{cf:DHV}.  In a nutshell, the disagreement lies on the fact that
the authors of \cite{cf:FLNO} predict that disorder is irrelevant at
marginality and, notably, that quenched and annealed critical points
coincide at small disorder, while the authors of \cite{cf:DHV} claim
that disorder is relevant for arbitrarily small disorder, leading to a
critical point shift of the order of $\exp( -c \gb^{-2})$ ($c>0$) for
$\gb \searrow 0$ ($\gb^2$ is the disorder variance).

Recently we have been able to prove that, at marginality, there is a
shift of the critical point induced by the presence of disorder
\cite{cf:GLTmarg}, at least for
Gaussian disorder. We have actually proven that the shift is at least
 $\exp( -c \gb^{-4})$. The purpose of the present work is to go
beyond \cite{cf:GLTmarg} in three aspects:
\begin{enumerate}
\item We want to deal with rather general disorder variables: we are going to
assume only that the exponential moments are finite.
\item We are going to improve the bound $\exp( -c \gb^{-b})$, $b=4$,
  on the critical point shift, to $b=2+\epsilon$ ($\epsilon>0$
  arbitrarily small, and $c=c(b)$).
\item We will prove our results for a generalized class of pinning
  models. Pinning models are based on discrete renewal processes,
  characterized by an inter-arrival distribution which has power-law
  decay (the exponent in the power law parametrizes the model and
  varying such parameter one explores the different types of critical
  behaviors we mentioned before). The generalized pinning model is
  obtained by relaxing the power law decay to  regularly varying decay, that is (in
  particular) we allow {\sl logarithmic correction} to power-law
  decay.  This, in a sense, allows zooming into the marginal case and
  makes clearer the interplay between the underlying renewal and the
  disorder variables.
\end{enumerate}

\subsection{The framework and some basic facts}
In mathematical terms, disordered pinning models are one-dimensional
Gibbs measures with random one-body potentials and reference measure
given by the law of a renewal process.  Namely, pinning models are
built starting from a (non-delayed, discrete) renewal process $\tau=
\{ \tau_n \}_{n=0,1, \ldots}$, that is a sequence of random variables
such that $\tau_0=0$ and $\{ \tau_{j+1}-\tau_j\}_{j=0,1, \ldots}$ are
independent and identically distributed with common law (called {\sl
  inter-arrival distribution}) concentrated on $\N:=\{1,2, \ldots\}$
(the law of $\tau$ is denoted by $\bP$): we will actually assume that
such a distribution is regularly varying of exponent $1+\ga$, {\sl
  i.e.}
\begin{equation}
\label{eq:K}
K(n) \, :=\, \bP(\tau_1=n)=\frac{L(n)}{n^{1+\ga}}, \ \ \text{ for } \ n=1,2, \ldots,
\end{equation}
where $\ga \ge 0$ and $L(\cdot)$ is a slowly varying function, that is $L: (0 , \infty) \to (0, \infty)$
is measurable and it satisfies $\lim_{x \to \infty}L(c x)/L(x)=1$ for every $c>0$. 
There is actually no loss of generality in assuming $L(\cdot)$ smooth
and we will do so  (we refer to %the self contained book
 \cite{cf:RegVar} for properties of slowly varying functions).

\smallskip
\begin{rem}
\label{rem:SVF}\rm
Examples of slowly varying functions include {\sl logarithmic slowly
  varying functions} (this is probably not a standard terminology, but
it will come handy), that is the positive measurable functions that
behave like $\consta (\log(x))^{\constb}$ as $x\to \infty$, with
$\consta>0$ and $\constb \in \R$. These functions are just a
particular class of slowly varying functions, but it is already rich
enough to appreciate the results we are going to present.  Moreover we
will say that $L(\cdot)$ is trivial if $\lim_{x \to \infty} L(x)=c \in
(0, \infty) $.  The general statements about slowly varying function
 that we are going to use can be verified in an elementary
way for logarithmic slowly varying functions; readers who feel
uneasy with the general theory may safely focus on this restricted
class.
\end{rem}
\smallskip

Without loss of generality we assume that $\sum_{n \in \N} K(n)=1$
(actually, we have implicitly done so when we have introduced $\tau$).
This does not look at all like an innocuous assumption at first,
because it means that $\tau$ is {\sl persistent}, namely $\tau_j<
\infty$ for every $j$, while if $\sum_n K(n) < 1$ then $\tau$ is {\sl
  terminating}, that is $\vert \{j : \, \tau_j < \infty\}\vert <
\infty$ a.s.. It is however really a harmless assumption, as explained
in detail in \cite[Ch.~1]{cf:Book} and recalled in the caption of
Figure~\ref{fig:RW}.

The disordered potentials  are introduced by means
of the IID sequence $\{ \go_n\}_{n=1,2, \ldots}$ of random variables
(the {\sl charges}) such that $\M( t) := \bbE[ \exp(t\go_1)]< \infty$
for every $t$. Without loss of generality we may and do assume 
that $\bbE[\go_1]=0$ and $\text{var}_\bbP(\go_1)=1$.

The model we are going to focus on 
is defined by the sequence of  probability measures
$\bP_{N,\go, \gb, h}= \bP_{N,\go}$, indexed by $N \in \N$, defined by
\begin{equation}
\label{eq:Gibbs}
\frac{\dd \bP_{N, \go}}{\dd \bP} (\tau) \, :=\, 
\frac1{Z_{N, \go}} \exp \left(\sum_{n=1}^N 
\left( \gb \go_n +h - \log \M(\gb) \right) \gd_n 
\right)  \gd_N \, ,
\end{equation}
where $\gb \ge 0$, $h \in \R$, $\gd_n $ is the indicator 
function that $n= \tau_j$ for some $j$ and 
$Z_{N , \go}$ is the partition function, that is the normalization constant.
It is practical to look at $\tau$ as a random subset of $\{0\} \cup \N$,
so that, for example, $\gd_n = \ind_{n \in \tau}$.

\smallskip

\begin{rem}
\rm
We have chosen $\M(t)< \infty$ for every $t$
only for ease of exposition. The results we present directly
generalize to the case in which $\M(t_0)+\M(-t_0)< \infty$ for
a $t_0>0$. In this case it suffices to look at the system
only for $\gb \in [0, t_0)$.  
\end{rem}
\smallskip

Three comments on \eqref{eq:Gibbs} are in order:
\smallskip

\begin{enumerate}
\item we have introduced the model in a very general set-up which is, possibly,
 not too  intuitive, but it  allows
a unified approach to a large class of models \cite{cf:Fisher,cf:Book}. It may
be useful at this stage to look at Figure~\ref{fig:RW} that
illustrates the random walk pinning model; 
\item the presence of $-\log \M (\gb)$ in the exponent is just 
a parametrization of the problem that comes particularly handy and it
can be absorbed by redefining $h$;
\item the presence of $\gd_N$ in the right-hand side means that we are looking
only at trajectories that are {\sl pinned} at the endpoint of the system. 
This is just a boundary condition and we may as well remove 
$\gd_N$ for the purpose of the results that we are going to state,
since it is well known for example that the free energy of this system
is independent of the boundary condition ({\sl e.g.} \cite[Ch.~4]{cf:Book}). Nonetheless, at a technical level
it is more practical to work with the system pinned at the endpoint.
\end{enumerate}

\smallskip

\begin{figure}[hlt]
\begin{center}
\leavevmode
\epsfxsize =14.5 cm
\psfragscanon
\psfrag{0}[c][l]{\small $0$}
\psfrag{Sn}[c][l]{\small $S_n$}
\psfrag{n}[c][l]{\small $n$}
\psfrag{t0}[c][l]{\small $\tau_0$}
\psfrag{t1}[c][l]{\small $\tau_1$}
\psfrag{t2}[c][l]{\small $\tau_2$}
\psfrag{t3}[c][l]{\small $\tau_3$}
\psfrag{t4}[c][l]{\small $\tau_4$}
\psfrag{o1}[c][l]{\small $\go_1$}
\psfrag{o2}[c][l]{\small $\go_2$}
\psfrag{o3}[c][l]{\small $\go_3$}
\psfrag{o4}[c][l]{\small $\go_4$}
\psfrag{o5}[c][l]{\small $\go_5$}
\psfrag{o6}[c][l]{\small $\go_6$}
\psfrag{o7}[c][l]{\small $\go_7$}
\psfrag{o8}[c][l]{\small $\go_8$}
\psfrag{o9}[c][l]{\small $\go_9$}
\psfrag{oa}[c][l]{\small $\go_{10}$}
\psfrag{ob}[c][l]{\small $\go_{11}$}
\psfrag{oc}[c][l]{\small $\go_{12}$}
\psfrag{od}[c][l]{\small $\go_{13}$}
\psfrag{oe}[c][l]{\small $\go_{14}$}
\psfrag{of}[c][l]{\small $\go_{15}$}
\epsfbox{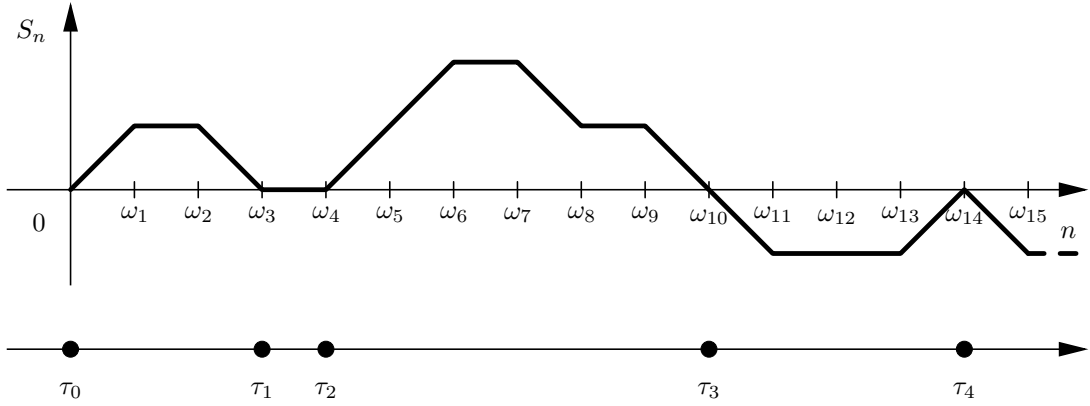}
\caption{\label{fig:RW} A symmetric random walk trajectory with
  increments taking values in $\{-1,0,+1\}$ is represented as a
  directed random walk. On the $x$-axis, the {\sl defect line}, there
  are quenched charges $\go$ that are collected by the walk when it
  hits the charge location.  The energy of a trajectory just depends
  on the underlying renewal process $\tau$. For the case in the
  figure, $K(n):=\bP(\tau_1=n)\sim const. n^{-3/2}$ for $n \to \infty$
  ({\sl e.g.} \cite[App.~A.6]{cf:Book}).  Moreover the walk is
  recurrent, so $\sum_n K(n)=1$.  There is however another
  interpretation of the model: the charges may be thought of as sticking
  to $S$, not viewed this time as a directed walk. If the walk hits
  the origin at time $n$, the energy is incremented by $(\beta\go_n+h-
\log \M(\gb))$.  This
  interpretation is particularly interesting for a three-dimensional
  symmetric walk in $\Z^3$: the walk may be interpreted as a polymer
  in $d=3$, carrying charges on each monomer, and the monomers
  interact with a point in space (the origin) via a charge-dependent
  potential.  Also in this case $K(n)\sim const. n^{-3/2}$, but the
  walk is transient so that $\sum_n K(n) <1$ ({\sl e.g.} 
  \cite[App.~A.6]{cf:Book}).  It is rather easy to see that any model based on a
  terminating renewal with inter-arrival distribution $K(\cdot)$ can
  be mapped to a model based on the persistent renewal with
  inter-arrival distribution $K(\cdot)/\sum_n K(n)$ at the expense of
  changing $h$ to $h+\log \sum_n K(n)$.  For much more detailed
  accounts on the (very many!) models that can be directly mapped to
  pinning models we refer to \cite{cf:Fisher,cf:Book}.  }
\end{center}
\end{figure}

The (Laplace) asymptotic behavior of $Z_{N , \go}$ shows a phase transition.
In fact, if we define the free energy as
\begin{equation}
\label{eq:fe}
\tf (\gb, h) \, :=\, \lim_{N \to \infty} \frac 1N \bbE \log Z_{N, \go},
\end{equation}
where the limit exists since the sequence $\{ \bbE \log Z_{N, \go} \}_N$
is super-additive (see {\sl e.g.} \cite[Ch.~4]{cf:Book}, where it is also
proven that $\tf (\gb, h)$ coincides with the $\bbP(\dd \go)$-almost sure
limit of $(1/N)\log Z_{N, \go}$, so that $\tf(\gb, h)$ is effectively
the {\sl quenched} free energy), then it is easy to see that 
$\tf(\gb, h)\ge 0$: in fact,
\begin{multline}
\tf(\gb, h)\, \ge \, \limsup_{N \to \infty}
\frac 1N \bbE \log 
\bE\left[
\exp \left(\sum_{n=1}^N 
\left( \gb \go_n +h - \log \M(\gb) \right) \gd_n 
\right)   \ind_{\tau_1=N}\right] 
\\ =\, \lim_{N \to \infty} \frac 1N \left( (h -\log \M(\gb))+ \log \bP(\tau_1=N)
\right)\, =\, 0.
\end{multline}
The transition we are after is captured by
setting
\begin{equation}
h_c(\gb) \, :=\, \sup \{ h:\, \tf(\gb,h)=0\} \, =\, \inf \{ h:\, \tf(\gb,h)>0\} ,
\end{equation}
where the equality is a direct consequence of the fact that
$\tf(\gb, \cdot)$ is non-decreasing (let us point out also that 
the free energy is  a continuous function of both arguments, as it
follows from standard convexity arguments). 
We have the bounds (see point (2) just below for the proof)
\begin{equation}
\label{eq:febounds}
\tf(0,h-\log \M (\gb)) \, \le \, \tf(\gb, h) \, \le \, \tf(0, h)\, , 
\end{equation}
which directly imply
\begin{equation}
\label{eq:hcbounds}
h_c(0) \, \le \, h_c (\gb) \, \le \, h_c(0)+\log \M(\gb)\, .
\end{equation}
Two important observations are:
\begin{enumerate}
\item the bounds in \eqref{eq:febounds} are given in terms of
$\tf(0, \cdot)$, that is the free energy of the non-disordered system, which 
can be solved analytically ({\sl e.g.} \cite{cf:Fisher,cf:Book}).
In particular $h_c(0)=0$
for every $\ga$ and every choice of $L(\cdot)$ (in fact
$h_c(0)= -\log \sum_n K(n)$ and we are assuming that $\tau$ is 
persistent).
We will keep in our formulae $h_c(0)$ both because
we think that it makes them more readable and because
they happen to be true also if $\tau$ were a terminating renewal).
\item The upper bound in \eqref{eq:febounds}, that entails the lower bound
in \eqref{eq:hcbounds},
follows directly from the standard {\sl annealed bound}, that is 
$\bbE \log Z_{N ,\go}\le  \log \bbE Z_{N ,\go}$, and by observing that
 the {\sl annealed partition function} $\bbE Z_{N ,\go}$ coincides with
 the partition function of the quenched model with $\gb=0$, that is simply the
 non-disordered case
 (of course, the presence of the term
 $-\log \M (\gb)$ in \eqref{eq:Gibbs} finds here its motivation).
 The lower bound in \eqref{eq:febounds}, entailing the upper bound
in \eqref{eq:hcbounds}, follows by a convexity argument too
(see \cite[Ch.~5]{cf:Book}).
\end{enumerate}
\smallskip

\begin{rem}\rm
It is rather easy (just take the derivative of the free energy
with respect to $h$) to realize that the phase transition we have
outlined in this model is a localization transition: when $h< h_c(\gb)$,
for $N$ large, the random set $\tau$ is {\sl almost empty},
while when $h> h_c(\gb)$ it is of size $const. N$ (in fact $const.= \partial_h\tf (\gb ,h)$).
Very sharp  results have been obtained on this issue: we refer to 
\cite[Ch.s~7 and 8]{cf:Book} and references therein. 
\end{rem}

\subsection{The Harris criterion}
\label{sec:Harris}
We can now make precise the Harris criterion predictions
mentioned in \S~\ref{sec:RIM}. As we have seen, in our case
 the pure (or {\sl annealed}) model is just
the non-disordered model, and the latter is exactly solvable, so that
the critical behavior is fully understood, notably \cite[Ch.~2]{cf:Book}
\begin{equation}
\label{eq:feexp}
\lim_{a \searrow 0}\frac{\log \tf(0, h_c(0)+a)}{\log a}\,=\,
\max \left( 1, \frac 1{\ga} \right)\, =:\, \nu_{\text{pure}}. 
\end{equation}  
The specific heat exponent of the pure model (that is the critical
exponent associated to %%F   
 $1/\partial_h^2 \tf (0, h)$) is
computed analogously and it is equal to $2-\nu_{\text{pure}}$.
Therefore the Harris criterion predicts {\sl disorder relevance} for
$\ga>1/2$ ($2-\nu_{\text{pure}}>0$) and {\sl disorder irrelevance} for
$\ga<1/2$ ($2-\nu_{\text{pure}}<0$) at least for $\gb$ below a
threshold, with $\ga=1/2$ as marginal case.  So, what one expects is
that $\nu_{\text{pure}}=\nu_{\text{quenched}}$ (with obvious
definition of the latter) if $\ga<1/2$ for $\gb$ {\sl not too large}
and $\nu_{\text{pure}}\neq\nu_{\text{quenched}}$ if $\ga>1/2$ (for
every $\gb>0$).

While a priori the Harris criterion attacks the issue of critical behavior,
it turns out that  a Harris-like approach
in the pinning context  \cite{cf:FLNO,cf:DHV} yields information also 
on $h_c(\gb)$, namely that $h_c(\gb)=h_c(0)$
if $\ga<1/2$ and $\gb$ again not too large, while 
$h_c(\gb)>h_c(0)$ as soon as $\gb>0$. For the sequel it is 
 important to recall some aspects of the approaches in  \cite{cf:FLNO,cf:DHV}.

\smallskip

The main focus of \cite{cf:FLNO,cf:DHV},
 is on the case $\ga=1/2$ and trivial $L(\cdot)$.
%%F
% (while in  \cite{cf:FLNO} the disorder is assumed to be Gaussian (?),
% in \cite{cf:DHV}
%  the disorder appears to be rather general). 
In fact they focus on
 the interface wetting problem in two dimensions, that boils down to
 directed random walk pinning in $(1+1)$-dimensions. In this framework
 the conclusions of the two papers differ:  \cite{cf:FLNO}
 stands for $h_c(\gb)=h_c(0)$ for $\gb $ small, while  in \cite{cf:DHV}
one finds an argument in favor of 
\begin{equation}
\label{eq:DHV}
h_c(\gb)-h_c(0) \approx \exp( -c \gb^{-2}),
\end{equation}
as $\gb \searrow 0$ (with $c>0$ an explicit constant).

We will not go into the details 
of these arguments, but we wish to point out why, in these arguments,  $\ga=1/2$
plays such a singular role. 
\smallskip

\begin{enumerate}
\item In the approach of \cite{cf:FLNO} an expansion of the free
  energy to all orders in the variance of $\exp(\gb\go_1 -\log
  \M(\gb))$, that is $(\M(2\gb)/\M^2(\gb))-1 \stackrel{\gb \searrow
    0}\sim \gb^2$, is performed. In particular (in the Gaussian case)
\begin{equation}
\label{eq:FLNO}
\tf(\gb, h_c(0) + a)\, =\ \tf(0, h_c(0) + a) - \frac 12 \left( \exp(\gb^2)-1
\right) \left( \partial_a \tf(0, h_c(0) + a) \right)^2 + \ldots
\end{equation}
and, when $L(\cdot)$ is trivial, 
$\partial_a \tf(0, h_c(0) + a) $ behaves like (a constant times)
%%F
$ a^{(1-\ga)/\ga}$ for $\ga \in (0,1)$ (this is detailed
for example in \cite{cf:GBreview}) 
%%F
and like a constant 
for $\ga\ge1$.  This suggests that the
expansion \eqref{eq:FLNO} cannot work for $\ga>1/2$, because the
second-order term, for $a \searrow 0$, becomes larger than the first
order term ($a^{\max(1/\ga, 1)}$).  The borderline case is $\ga=1/2$,
and trust in such an expansion for $\ga =1/2$ may follow from the fact
that $\gb$ can be chosen small. 
%%F
In conclusion, an argument along the lines of \cite{cf:FLNO} predicts disorder
relevance if and only if $\ga>1/2$ (if $L(\cdot)$ is trivial).
%%%
\item The approach of \cite{cf:DHV} instead is based on the analysis
  of $\text{var}_{\bbP} (Z_{N, \go})$ at the pure critical point
  $h_c(0)$.  This directly leads to studying the random set $\tilde
  \tau:= \tau \cap \tau^\prime$ (it appears in the computation in a
  very natural way, we call it {\sl intersection renewal}), with
  $\tau^\prime$ an independent copy of $\tau$ (note that $ \tilde \tau
  $ is still a renewal process): in physical terms, one is looking at
  the {\sl two-replica system}.  It turns out that, even if we have
  assumed $\tau$ persistent, $ \tilde \tau$ may not be: in fact, if
  $L(\cdot)$ is trivial, then $\tilde \tau$ is persistent if and only
  if $\ga \ge 1/2$ (see just below for a proof of this fact). And
  \cite{cf:DHV} predicts disorder relevance if and only if $\ga\ge
  1/2$.
\end{enumerate}
\smallskip
Some aspects of these two approaches were made rigorous mathematically: The expansion of the free energy \eqref{eq:FLNO} was proved to hold for $\alpha<1/2$ in \cite{cf:GTirrel}, and the second moment analysis of \cite{cf:DHV} was used to prove {\sl disorder irrelevance} in \cite{cf:Ken,cf:T_cmp}, making it difficult to choose between the predictions.
\smallskip

We can actually find in the physical literature a number of authors standing for one or the other
of the two predictions in the marginal case $\ga=1/2$ (the reader can
find a detailed review of the literature in  \cite{cf:GLTmarg}).
 But we would like to go a step 
farther and we point out that,
by generalizing naively the  approach in  \cite{cf:DHV}, one is  tempted 
to conjecture disorder relevance (at arbitrarily small $\gb$)
if and only if the intersection renewal is recurrent. Let us 
 make this condition explicit: while one does not have direct access
to the inter-arrival distribution of $\tilde \tau$, it is straightforward, by independence, 
to write the renewal function of $\tilde \tau$:
\begin{equation}
\bP(n \in \tilde \tau) \, =\, \bP( n \in \tau)^2.
\end{equation}
It is then sufficient to use the basic (and general) renewal process
formula $\sum_n \bP( n \in \tilde \tau) = (1- \sum_n \bP( \tilde
\tau_1=n))^{-1}$ to realize that $\tilde \tau$ is persistent if and
only if $\sum_n \bP( n \in \tilde \tau)=\infty$.  Since under our
assumptions for 
%%F
$\ga\in(0,1)$ \cite[Th.~B]{cf:Doney}
\begin{equation}
\label{eq:Doney}
\bP(n \in \tau) \stackrel{n \to \infty} \sim \frac{\ga \sin (\pi \ga)}{\pi} \frac1{n^{1-\ga}L(n)},
\end{equation} 
we easily see that the intersection renewal $\tilde \tau$ is persistent
for $\ga >1/2$ and terminating if $\ga<1/2$ (the case $\ga=0$ can be treated too \cite{cf:RegVar},
and $\tilde \tau$ is terminating).
In  the $\ga=1/2$ case the argument we have
just outlined yields
\begin{equation}
\label{eq:persist}
\tau\cap \tau^\prime \ \text{ is persistent } \ \Longleftrightarrow  \ 
\sum_n \frac 1{n \, L(n)^2} \, = \, \infty.
\end{equation}
Roughly, this is telling us that the intersection renewal $\tilde
\tau$ is persistent up to a slowly varying function $L(x)$ diverging
{\sl slightly less} than $(\log x)^{1/2}$.  In particular, as we have
already pointed out, if $L(\cdot)$ is trivial, $\tilde \tau$ is
persistent.

Let us remark that
the expansion
\eqref{eq:FLNO} has been actually made rigorous in \cite{cf:GTirrel}, but 
only under the assumption that the intersection renewal $\tilde \tau$
is terminating (that is,  $\constb>1/2$ for logarithmic slowly varying functions).

\smallskip

\begin{rem}
\label{rem:Ltilde}
\rm
In view of the argument we have just outlined, we introduce 
the increasing function  $\tilde L : (0, \infty) \to (0, \infty)$ defined as
\begin{equation}
\label{eq:Ltilde}
\tilde L (x) \, :=\, \int_0^x \frac1{(1+y)L(y)^2}\dd y,
\end{equation}
that is going to play a central role from now on.
Let us point out that, by \cite[Th.~1.5.9a]{cf:RegVar}, $\tilde L(\cdot)$ is a slowly varying function 
which has the property
\begin{equation}
\label{eq:Ltilde-prop}
\lim_{x \to \infty}\tilde L(x) L(x)^2\, =\, +\infty,
\end{equation}
which is a non-trivial statement when $L(\cdot)$ does not diverge at infinity. 
Of course we are most interested in the fact that, 
when $\ga=1/2$,  $\tilde L(x)$ diverges as $x \to \infty$
 if and only if the intersection renewal $\tilde \tau$ is recurrent ({\sl cf.}
 \eqref{eq:persist}).
 For completeness we point out  that 
$\tilde L(\cdot)$ is  a special type of slowly varying function
(a {\sl den Haan function} \cite[Ch.~3]{cf:RegVar}), but we will not exploit
the further regularity properties stemming out of this observation. 
%%F faut-il laisser cette remarque sur deen HAan?
%in the class $\Pi_{1/L(\cdot)^2}$, because
%\begin{equation}
%\lim_{x \to \infty}\frac{\tilde L(\gl x) - \tilde L (x) }{1/L(x)^{2}}\, =\, \log \gl \, ,
%\end{equation}
%and this gives further (useful?) properties. 
\end{rem}

\smallskip

%\begin{rem}
%\label{rem:FLN0}
%\rm
%One can generalize also the arguments in \cite{cf:FLNO} to non-trivial
%$L(\cdot)$ and one obtains that if
% $\tf(0, h_c(0) + a)=:b$ is small, then  $\partial_a \tf(0, h_c(0) + a) $
%behaves like $b^{1/2}/L(1/b)$ (times a constant). This goes in the direction
%of saying that, by taking the  \cite{cf:FLNO} approach, disorder is irrelevant
%for $\ga=1/2$
%as soon as $\liminf_{x \to \infty}L(x)>0$, that is, for logarithmic slowly varying functions
%(Remark~\ref{rem:SVF}), as soon as $\constb \ge 0$.
%And, in a sense, this argument may look rather convincing when
%$L(\cdot)$ diverges ($\constb>0$).   
%\end{rem}

\subsection{Review of the rigorous  results}
\label{sec:review-rs}
Much mathematical work has been done on disordered pinning models
recently. Let us start with a quick review of the $\ga \neq 1/2$ case:
\begin{itemize}
\item If $\ga>1/2$ disorder relevance is established. The positivity of 
$h_c(\gb)-h_c(0)$ (with precise asymptotic estimates as $\gb \searrow 0$)
 is proven \cite{cf:DGLT,cf:AZ}. It has been also shown that disorder has a smoothing
 effect on the transition and the quenched free energy critical exponent differs
 from the annealed one
  \cite{cf:GT_cmp}.
\item If $\ga<1/2$ disorder irrelevance is established, along with 
a number of sharp results saying in particular that, if $\gb$ is not too large, 
$h_c(\gb)=h_c(0)$
and that the free energy critical behavior coincides in the quenched and annealed
framework \cite{cf:Ken,cf:T_cmp,cf:GTirrel,cf:AZ2}.
\end{itemize}
\medskip

In the case $\ga=1/2$ results are less complete. 
Particularly relevant for the sequel are the next two results that we state as theorems.
The first one is taken from \cite{cf:Ken} (see also \cite{cf:GT_cmp})
and uses the auxiliary function $a_0(\cdot)$ defined by
\begin{equation}
a_0(\gb)\, :=\, C_1 L\left( \tilde L^{-1}\left(C_2/\gb^2\right) \right) 
\Big / \left( \tilde L^{-1}\left(C_2/\gb^2\right) \right) ^{1/2}\,  \ \ \text{ with }C_1>0 \text{ and }
C_2>0,
\end{equation}
if $\lim_{x \to \infty}\tilde L(x)=\infty$, and 
$a_0(\cdot)\equiv 0$ otherwise.

\medskip

\begin{theorem}
\label{th:1/2UB}
Fix $\go_1 \sim \cN(0,1)$, $\ga=1/2$ and choose a slowly varying
function $L(\cdot)$.  Then there exists $\gb_0>0$ and $a_1>0$ such
that for every $\epsilon>0$ there exist $C_1$ and $C_2>0$ such that
\begin{equation}
1-\epsilon \, \le \, 
\frac{\tf(\gb, a) }{\tf(0, a)}\, \le \, 1 \ \ \text { for } \ a> a_0(\gb), \, a \le a_1
\text { and } \gb\le \gb_0. 
\end{equation}
This implies for $\gb \le \gb_0$
\begin{equation}
\label{eq:hcUB}
h_c(\gb)- h_c(0) \, \le \, a_0(\gb).
\end{equation}
\end{theorem}
\medskip

It is worth pointing out that Theorem~\ref{th:1/2UB} yields
an upper bound matching  \eqref{eq:DHV} when $L(\cdot)$
is trivial. 

The next result addresses instead the lower bound on 
$h_c(\gb)-h_c(0)$ and it is taken from \cite{cf:GLTmarg}:

\medskip

\begin{theorem}
\label{th:1/2LB}
Fix $\go_1 \sim \cN(0,1)$ and  $\ga=1/2$.
If $L(\cdot) $ is trivial, then 
$h_c(\gb)-h_c(0)>0$ for every $\gb>0$ and 
there exists $C>0$ such that
\begin{equation}
\label{eq:beta4}
h_c(\gb)-h_c(0)\, \ge \, \exp\left(-C/ \gb^4 \right),
\end{equation}
for $\gb \le 1$.
\end{theorem}
\medskip

It should be pointed out that \cite{cf:GLTmarg} has been worked out
for trivial $L(\cdot)$, addressing thus precisely the controversial
issue in the physical literature. 
%%F
%, but the approach therein can be
%generalized to cover the case $\limsup _{x \to \infty} L(x) < \infty$
%and also some cases in which $L(\cdot)$ diverges.  
The case of $\lim_{x \to
  \infty}L(x)=0$ has been treated \cite{cf:AZ} (see \cite{cf:DGLT} for
a weaker result) where $h_c(\gb)-h_c(0)>0$ has been established with
an explicit but not optimal bound.
We point out also that a result analogous to Theorem~\ref{th:1/2LB}
has been proven  for a hierarchical version
of the pinning model (see
\cite{cf:GLTmarg} for the case of the
  hierarchical model proposed in \cite{cf:DHV}). 
\medskip

The understanding of the marginal case is therefore 
still partial and the following problems are clearly open:
\smallskip

\begin{enumerate}
\item What is really the behavior of $h_c(\gb)-h_c(0)$ in the marginal case? In particular,
for $L(\cdot)$ trivial, is \eqref{eq:DHV} correct?
\item Going beyond the case of $L(\cdot)$ trivial: is the 
 two-replica condition \eqref{eq:persist}
 equivalent to disorder relevance for small $\gb$?
\item What about non-Gaussian disorder? It should be pointed out that
  a part of the literature focuses on Gaussian disorder, notably
  Theorem~\ref{th:1/2UB}, but this choice appears to have been made in
  order to have more concise proofs (for example, the results in
  \cite{cf:DGLT} are given for very general disorder distribution).
  Theorem~\ref{th:1/2LB} instead exploits a technique that is more
  inherently Gaussian and generalizing the approach in
  \cite{cf:GLTmarg} to non-Gaussian disorder is not straightforward.
\end{enumerate}

\smallskip

As we explain in the next subsection, in this paper 
we will give {\sl almost} complete answers to 
questions (1), (2) and (3). In addition we will prove a monotonicity result for the phase diagram of pinning model which holds in great generality.

\subsection{The main result}
Our main result requires the existence of $\epsilon\in (0, 1/2]$ such that
\begin{equation}
\label{eq:Lassumption}
L(x) \, =\, o\left( (\log (x))^{(1/2)-\epsilon}\right) \ \ \text{ as } \ x \to \infty,
\end{equation}
that is $ \lim_{x \to \infty} L(x) (\log (x))^{-(1/2)+\epsilon}=0$. Of
course, if $L(\cdot)$ vanishes at infinity, \eqref{eq:Lassumption}
holds with $\epsilon=1/2$.  Going back to the slowly varying function
$\tilde L(\cdot)$, {\sl cf.} Remark~\ref{rem:Ltilde}, we note that,
under assumption \eqref{eq:Lassumption}, we have
\begin{equation}
\label{eq:Lmore}
\tilde L(x) \stackrel{x \to \infty} {\gg} \int_2^x \frac 1{y ( \log y)^{1-2\epsilon}} \dd y
\, =\, \frac 1{2 \epsilon } (\log x)^{2 \epsilon} - \frac 1{2 \epsilon } (\log 2)^{2 \epsilon}.
\end{equation} 
Therefore, under assumption \eqref{eq:Lassumption},
we have that if $q>(2\epsilon)^{-1}$ then
\begin{equation}
\lim_{x \to \infty} \frac{\tilde L(x)}{L(x)^{2/(q-1)}}\, =\, \infty,
\end{equation}
which guarantees that given $q>(2\epsilon)^{-1}$ (actually, in the sequel
$q \in \N$) and $A>0$,
\begin{equation}
\label{eq:Delta}
\gD (\gb;q,A)\, :=\, \left( \inf\left\{ n\in \N :\, {\tilde L(n)}/{L(n)^{2/(q-1)}}\ge A \gb^{-2q/(q-1)}
  \right\} \right)^{-1}\,  
\end{equation}
is greater than $0$ for every $\gb>0$.
\smallskip

Our main result is

\medskip

\begin{theorem}
\label{th:main}
Let us assume that $\ga=1/2$ and that \eqref{eq:Lassumption}
holds for some $\epsilon\in (0, 1/2]$. 
For every $\gb_0$ and every integer $q>(2\epsilon)^{-1}$
there exists $A>0$ such that
\begin{equation}
h_c(\gb) -h_c(0) \, \ge \, \gD (\gb;q,A)\, >\, 0, 
\end{equation}
for every $\gb\le \gb_0$.
\end{theorem}

\medskip

The result may be more directly appreciated in the particular case of
$L(\cdot)$ of logarithmic type, {\sl cf.} Remark~\ref{rem:SVF}, with
$\constb<1/2$, so that \eqref{eq:Lassumption} holds with
$\epsilon<\min((1/2)-\constb,1/2)$.  By explicit integration we see
that $\tilde L(x) \sim (\consta^2(1-2\constb))^{-1} (\log
(x))^{1-2\constb}$ so that
\begin{equation}
 \frac{\tilde L(x)}{L(x)^{2/(q-1)}} \, \sim \, \frac {\consta^{-2q/(q-1)}}{(1-2\constb)}
 (\log (x))^{1-2\constb q(q-1)^{-1}}
\end{equation}
and in this case
\begin{equation}
\gD (\gb; q, A) \stackrel{\gb\searrow 0}\sim
\exp\left(-c(\constb,A,q) \gb^{-b}\right),
\end{equation}
where 
$c(\constb,A,q):=((1-2\constb)\consta^{2q/(q-1)}A)^{1/C}$
and $b:= 2q/((q-1)C)$ with 
$C:=1-2\constb q(q-1)^{-1}$.
In short, by choosing $q$ large the exponent
$b>2/(1-2\constb)$ becomes arbitrarily close to $2/(1-2\constb)$,
at the expense of course of a large constant
$c(\constb,A,q)$, 
%%F
since $A$ will have to be chosen sufficiently large.

We sum up these steps into the following simplified version
of Theorem~\ref{th:main}
\medskip

\begin{cor}
If $\ga=1/2$ and $L(\cdot) $ is of logarithmic type with $\constb\in (-\infty, 1/2)$
({\sl cf.} Remark~\ref{rem:SVF}) 
then $h_c(\gb)>h_c(0)$ for every $\gb>0$ and for every
$b>2/(1-2\constb)$
there exists
$c>0$ such that, for $\gb$ sufficiently small
\begin{equation}
h_c(\gb)-h_c(0)\, \ge \, \exp\left(-c \gb^{-b}\right).
\end{equation}
\end{cor}

\medskip

This result of course has to be compared with the 
 upper bound in Theorem~\ref{th:1/2UB} that for
 $L(\cdot)$ of logarithmic type yields for $\constb< 1/2$
\begin{equation}
h_c(\gb)-h_c(0)\, \le \, 
\tilde C_1 \gb^{-2 \constb/(1-2\constb)}\exp\left(-\tilde C_2 \gb^{-2 /(1-2\constb)}\right),
\end{equation}
where $\tilde C_1$ and $\tilde C_2$ are positive constants that depend
(explicitly) on $\consta$, $\constb$ and on the two constants $C_1$
and $C_2$ of Theorem~\ref{th:1/2UB} (we stress that $\tilde C_1>0$ and
$\tilde C_2>0$ for every $\consta>0$ and $\constb< 1/2$).  \smallskip

\smallskip

The main body of the proof of Theorem~\ref{th:main} is given in the
next section.  In the subsequent sections a number of technical
results are proven.  In the last section
(Section~\ref{sec:monotonicity}) we prove a general result (Proposition~\ref{th:monotonicity})
for the
models we are considering: the monotonicity of the free energy with
respect to $\gb$. This result, proven for other disordered models,
appears not to have been pointed out up to now for the pinning model.
We stress that Proposition~\ref{th:monotonicity}  %Section~\ref{sec:monotonicity}  %GBGB
 is  not used  in the rest of the
paper, but, as  discussed in Section~\ref{sec:monotonicity},
one can find a link of some interest with our main results.  

\section{Coarse graining, fractional moment and measure change arguments}
\label{sec:CGfmmc}

The purpose of this section is to reduce  the proof to 
a number of technical statements, that are going to be proven
in the next sections. In doing so, we are going to introduce the quantities and 
notations used in the technical statements and, at the same time,
we will stress the main ideas and the novelties with respect to earlier
approaches (notably, with respect to \cite{cf:GLTmarg}).

We anticipate that the main ingredients of the proof are (like in
\cite{cf:GLTmarg}) a coarse graining procedure and a fractional moment
estimate on the partition function combined with a change of measure.
However:
\begin{enumerate}
\item In \cite{cf:GLTmarg} we have exploited the Gaussian character of
  the disorder to introduce {\sl weak, long-range correlations} while
  keeping the Gaussian character of the random variables. In fact, the
  change of measure is given by a density that is just the exponential
  of a quadratic functional of $\go$, that is a measure change via a
  {\sl 2-body potential}.  In order to lower the exponent $4$ in the
  right-hand side of \eqref{eq:beta4} we will use $q$-body potentials
  $q=3,4,  \ldots$ (this is the $q$ appearing in
  Theorem~\ref{th:main}). Such potentials carry with themselves a
  number of difficulties: for example, when the law of the disorder is Gaussian, the
  modified measure is not. As a matter of fact, there are even
  problems in defining the modified disorder variables if one modifies
  in a straightforward way the procedure in \cite{cf:GLTmarg} to use
  $q$-body potentials, due to integrability issues: such problems may
  look absent if one deals with bounded $\go$ variables, but they
  actually reappear when taking limits. The change-of-measure
  procedure is therefore performed by introducing $q$-body potentials
  {\sl and} suitable cut-offs.  Estimating the effect of
  such {\sl $q$-body potential with cut-off} change of measure is at
  the heart of our technical estimates.
\item The coarse-graining procedure is different from the one used in
  \cite{cf:T_cg,cf:GLTmarg}, since we have to adapt it to the new
  change of measure procedure. However, unlike point (1), 
  the difference between the previous coarse graining procedure and the
  one we are employing now is more technical than conceptual. 
\end{enumerate}

\subsection{The coarse graining length}
Recall the definition \eqref{eq:Ltilde} of $\tilde L(\cdot)$.
We are assuming  \eqref{eq:Lassumption}, therefore $\lim_{x \to \infty}\tilde
L(x)=+\infty$.
Chosen a value of  $q \in\{ 2,3, \ldots\}$ ($q$ is kept fixed throughout the proof) and a positive constant $A$ (that is going to be chosen large) we define
\begin{equation}
\label{eq:k}
k\, =\, k(\gb; q,A)\, :=\, 
\inf \left\{ n \in \N :\, \tilde L(n)/ L(n)^{2/(q-1)} \ge A \gb^{-2q/(q-1)}\right\}.
\end{equation}
Since we are interested also in cases in which  $L(\cdot)$ diverges (and possibly faster than
$\tilde L(\cdot)$) it is in general false  that $k< \infty$. However, the 
assumption \eqref{eq:Lassumption} guarantees that, for $q>(2\epsilon)^{-1}$,
 $L(x)/ L(x)^{2/(q-1)}\to \infty$ for $x\to \infty$ and therefore $k< \infty$.
 
Moreover, if $L(\cdot)$ is of logarithmic type (Remark~\ref{rem:SVF}) with
$\constb<1/2$,
then   for $q>1/(1-2\constb)$ the function $\tilde L(\cdot)/ L(\cdot)^{2/(q-1)}$
is (eventually) increasing.

\smallskip

Of course $k(\gb; q,A)$ is just $1/\gD (\gb;q,A)$,
{\sl cf.}
\eqref{eq:Delta}, and the reason for such a link %, as well as the reason for defining $k$ as in \eqref{eq:k},
is explained in Remark~\ref{rem:kgD}.
Note by now  that $k$ is monotonic in both $\gb$ and $A$.
Since $\gb$ is chosen smaller than an arbitrary fixed quantity
$\gb_0$, in order to guarantee that $k$ is large we will rather play on 
choosing $A$ large.

\smallskip

\begin{rem}
\label{rem:L}
\rm
For the proof  certain monotonicity properties will be important. Notably, 
we know \cite[\S~1.5.2]{cf:RegVar} that $1/(\sqrt{x} L(x))$ is asymptotic 
to a monotonic (decreasing) function and this directly implies that we can
find a slowly varying function $\bL(\cdot)$ and a constant $\const _L \in (0,1]$ such that
\begin{equation}
\label{eq:Lb}
x \mapsto \frac1{\sqrt{x}\bL(x)} \text{ is decreasing  and } 
 \const _L  \bL(x) \, \le \, L(x) \, \le \, \bL(x) \text{ for every } x\in (0, \infty).
\end{equation}
Given the asymptotic behavior of the renewal function of $\tau$ (a special case
of \eqref{eq:Doney})
\begin{equation}
\label{eq:Doney1/2}
\bP\left(n \in \tau\right)
\stackrel{n \to \infty}\sim \frac 1{2\pi \sqrt{n} L(n)},
\end{equation}
and the fact that $\bP\left(n \in \tau\right)>0$ for every $n\in \N$, 
 we can choose
$\bL (\cdot)$ and 
$\const _L $ such that we have also
\begin{equation}
\label{eq:Doney-bound}
\frac{1}{\sqrt{n+1}\, \bL(n+1)}\, \le \,
\bP(n \in \tau) \, \le \, \frac{\const _L ^{-1}}{\sqrt{n+1}\, \bL(n+1)}, \ \ \ n=0,1,2, \ldots .
\end{equation}
It is natural to choose $\bL(\cdot)$ such that $\lim_{x \to \infty}\bL(x)/L(x)
\in [1, 1/\const_L)$ exists, and we will do so.
For later convenience
we set
\begin{equation}
\label{eq:R12}
R_{\frac12} (x) \, :=\,  \frac1{\sqrt{x+1}\, \bL(x+1)}.
\end{equation}
\end{rem}

\subsection{The coarse graining procedure and the fractional moment bound}
Let us start by introducing for $0\le M < N$ the notation(s)
\begin{equation}
  \label{eq:Znh}
  Z_{M, N}\, =\,
  Z_{M, N,\go}:=\bE\left[e^{\sum_{n=M+1}^N(\beta\go_n+h-\log \M(\gb))\delta_n}\delta_N
  \big \vert \gd_M=1\right]\, ,
\end{equation}
%%F
and $Z_{M,M}:=1$ (of course $Z_{N, \go}= Z_{0, N}$).
We consider without loss of generality 
a system of size proportional to $k$, that is $N=km$
with $m \in \N$. For $\mathcal I\subset
\left\{1,\dots,m\right\}$ we define
\begin{equation}
  \hat Z_{\go}^{\mathcal I}:=\bE \left[e^{\sum_{n=1}^N(\beta\go_n+h-\log \M(\gb))\delta_n} \gd_N \ind_{E_\cI}(\tau)\right],
\end{equation}
where
 $E_\cI:=\{ \tau \cap (\cup_{i \in \cI} B_i) = \tau \setminus \{0\}\}$, 
and
\begin{equation}
 B_i:=\left\{(i-1)k+1,\dots,ik\right\},
\end{equation}
that is $E_\cI$ is the event that the renewal $\tau$ intersects the
blocks $(B_i)_{i\in \mathcal I}$ and only these blocks over $\{1,
\ldots, N\}$.  It follows from this definition that
\begin{equation}\label{eq:decomp}
 Z_{N,\go}=\sum_{\mathcal I\subset \left\{1,\dots,m\right\}} \hat Z_{\go}^{\mathcal I}.
\end{equation}
Note that $\hat Z_{\go}^{\mathcal I}=0$ if $m\notin \mathcal I$. Therefore in the following we will always assume $m\in \mathcal I$.
For $\mathcal I=\{i_1,\dots, i_l\}$, ($i_1<\dots<i_l$, $i_l=m$), one can express $\hat Z_{\go}^{\mathcal I}$ in the following way:
\begin{multline}
\label{eq:decomp1}
\hat Z_{\go}^{\mathcal I}=
\sumtwo{d_1, f_1 \in B_{i_1}}{d_1\le f_1}\sumtwo{d_2, f_2\in B_{i_2}}{d_2\le f_2}\ldots \sum_{d_l\in B_{i_l}}\\
K(d_1)z_{d_1}Z_{d_1,f_1}K(d_2-f_1)Z_{d_2,f_2}\ldots K(d_l-f_{l-1})z_{d_l}
Z_{d_l,N},
 \end{multline}
 with $z_n := \exp( \gb \go_n +h - \log \M(\gb))$.  Let us fix a value of $\gamma\in (0,1)$
  (we  actually choose $\gamma=6/7$, but we will keep
  writing it as $\gamma$). Using the inequality $\left(\sum a_i\right)^{\gamma}\le
 \sum a_i^{\gamma}$ (which is valid for $a_i\ge 0$ and an arbitrary
 collection of indexes) we get
\begin{equation}
\label{eq:cg+fm}
  \bbE\left[ Z_{N,\go}^{\gamma}\right]\, \le\, \sum_{\mathcal I\subset \left\{1,\dots,m\right\}} \bbE \left[\left(\hat Z_{\go}^{\mathcal I}\right)^{\gamma}\right].
\end{equation}
An elementary, but crucial, observation is that
\begin{equation}
\label{eq:fmm}
\tf(\gb, h)\, =\, \lim_{N \to \infty} \frac1{\gamma N} \bbE
\log Z_{N, \go }^\gamma\, \le \, 
 \liminf_{N \to \infty} \frac1{\gamma N} 
\log \bbE Z_{N, \go }^\gamma,
\end{equation}
so that if we can prove that $\limsup_N\bbE Z_{N, \go }^\gamma < \infty$
%%F
for $h=h_c(0)+\Delta(\beta;q,A)$
we are done.

\begin{figure}[hlt]
\begin{center}
\leavevmode
\epsfxsize =14.5 cm
\psfragscanon
\psfrag{O}[c][l]{$0$}
\psfrag{1K}[c][l]{$k$}
\psfrag{2K}[c][l]{$2k$}
\psfrag{3K}[c][l]{$3k$}
\psfrag{4K}[c][l]{$4k$}
\psfrag{5K}[c][l]{$5k$}
\psfrag{6K}[c][l]{$6k$}
\psfrag{7K}[c][l]{$7k$}
\psfrag{8K}[c][l]{$8k=N$}
\psfrag{d1}[c][l]{\Small{$d_1$}}
\psfrag{d2}[c][l]{\Small{$d_2$}}
\psfrag{d3}[c][l]{\Small{$d_3$}}
\psfrag{d4}[c][l]{\Small{$d_4$}}
\psfrag{f1}[c][l]{\Small{$f_1$}}
\psfrag{f2}[c][l]{\Small{$f_2$}}
\psfrag{f3}[c][l]{\Small{$f_3$}}
\psfrag{f4N}[c][l]{\Small{$f_4=N$}}
\epsfbox{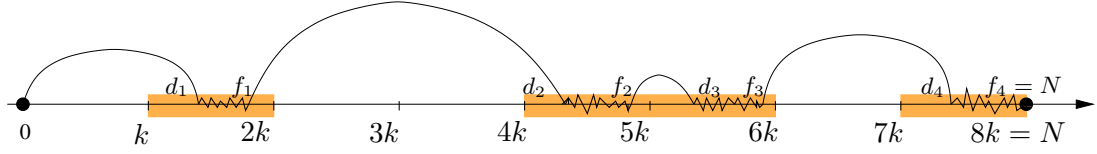}
\caption{\label{fig:cg} The figure above explains our coarse graining
  procedure. Here $N=8k$, $\mathcal I=\{2,5,6,8\}$. The drawn
  trajectory is a typical trajectory contributing to $\hat Z^{\mathcal
    I}_{N,\go}$; $d_i$ and $f_i$, $1\le i\le 4$, correspond to the
  indexes of \eqref{eq:decomp1}. The shadowed regions represent the sites
  on which the change of measure 
  procedure (presented in \S~\ref{sec:com}) acts.}
\end{center}
\end{figure}

\subsection{The change of measure}
\label{sec:com}
We introduce 
\begin{equation}
\label{eq:X}
 X_j\, :=\, \sum_{\ui \in B_j^q} V_k({\ui})\go_{\ui}, 
\end{equation}
where $ B_j^q$ is the Cartesian product of $B_j$ with itself $q$ times and 
$\go_{\ui}= \prod_{a=1}^q \go_{i_a}$. The {\sl potential}
 $V_k(\cdot)$ plays a crucial role for the sequel: we define it and discuss some
 of its
 properties  in the next remark.
 \medskip
 
 \begin{rem}
\label{rem:V}\rm
The potential $V$ is best introduced if we define the sorting operator $\sort(\cdot)$:
if $\ui  \in \R^q$ ($q =2,3, \ldots$), $\sort (\ui) \in \R^q$ is the non-decreasing rearrangement 
of the entries of $\ui$. We introduce then
\begin{equation}
\label{eq:U}
U({\ui}) \, :=\, \prod_{a=2}^{ q}
 R_{\frac12} \left( \sort(\ui)_a- \sort(\ui)_{a-1}\right),
\end{equation}
The {\sl potential} $V$ is defined  by renormalizing  $U$ and by setting to zero
the {\sl diagonal} terms: 
\begin{equation}
\label{eq:useR2}
V_k({\ui}) \, := \, \frac{1}{(q!)^{1/2}  k^{1/2}\tilde \bL(k)^{(q-1)/2}}U( \ui)
 \ind_{\{i_a\neq i_b \text{ for every } a, b\}}
,
\end{equation}
where $\tilde \bL(\cdot)$ is defined as in \eqref{eq:Ltilde}, with $L(\cdot)$
replaced by $\bL(\cdot)$. By exploiting the fact that for every $c>0$ 
we have $\sum_{i\le cN} R_{1/2}(i)^2 \stackrel{N \to \infty} \sim 
\tilde \bL (N)$
one  sees that 
\begin{equation}
\label{eq:cV-1}
\sum_{\ui \in B_1^q} V_k({\ui})^2 \, = \, \frac{1}{k\,  ( \tilde \bL(k))^{q-1}}
\sum_{0<i_1<\ldots < i_q \le k} \prod_{a=2}^q \left(R_{1/2}(i_a-i_{a-1})\right)^2
 \stackrel{k \to \infty}\sim  1.
\end{equation}
Therefore 
\begin{equation}
\label{eq:cV}
\sum_{\ui \in B_1^q} V_k({\ui})^2 \, \le \, 2,
\end{equation}
for $k$ sufficiently large.
\end{rem}

\medskip

Let us introduce, for $K>0$, also
\begin{equation}
\begin{split}
 f_K(x)\, &:=\, -K\ind_{\{x\ge \exp(K^2)\}},\\
 g_{\mathcal I}(\go)\, &:= \, \exp\left(\sum_{j\in \mathcal I}f_K(X_j)\right),\\
 \bar g(\go)\, &:= \, \exp(f_K(X_1)).
\end{split}
\end{equation}
We are now going to replace, for fixed $\cI$, the measure
$\bbP(\dd \go)$ with $ g_{\mathcal I}(\go) \bbP(\dd \go)$. The latter is not a probability measure:
we could normalize it, but this is inessential because we are directly exploiting 
 H\"older inequality to get
\begin{equation}
\label{eq:Holder}
\bbE \left[\left(\hat Z_{\go}^{\mathcal I}\right)^{\gamma}\right]
\, \le\,  
\left(\bbE\left[g_{\mathcal I}(\go)^{-\frac{\gamma}{1-\gamma}}\right]\right)^{1-\gamma}\left(\bbE\left[g_{\mathcal I}(\go)
\hat Z_{\go}^{\mathcal I}\right]\right)^{\gamma}.
\end{equation}
The first factor in the right-hand side is easily controlled, in fact
\begin{equation}
\label{eq:factor1}
\bbE\left[g_{\mathcal I}(\go)^{-\frac{\gamma}{1-\gamma}}\right]\,=
\, \bbE \left[ \bar g (\go)^{-\frac{\gamma}{1-\gamma}}\right]^{\vert \cI \vert} 
\, =\, \left[ \left(\exp\left(\frac{K\gamma}{1-\gamma}\right)-1\right)
\bbP\left(X_1\ge \exp\left(K^2\right)\right)+1 \right]^{\vert \cI \vert} ,
\end{equation}
and since $X_1$ is centered and its variance coincides 
with the left-hand side of \eqref{eq:cV}, by Chebyshev inequality
the term $ \exp\left({K\gamma}/{(1-\gamma)}\right)\bbP\left(X_1\ge \exp\left(K^2\right)\right)$
can be made arbitrarily small by choosing $K$ large. Therefore for $K$ sufficiently large
(depending only on $\gamma(=6/7)$)
\begin{equation}
\label{eq:fHolder}
\bbE \left[\left(\hat Z_{\go}^{\mathcal I}\right)^{\gamma}\right]
\, \le\,  2^{\gamma \vert \cI\vert}
\left(\bbE\left[g_{\mathcal I}(\go)
\hat Z_{\go}^{\mathcal I}\right]\right)^{\gamma}.
\end{equation}

Estimating the remaining factor is a more involved matter. We will actually use the following two statements,
that we prove in the next section. Set
$P_\cI \, :=\, 
 \bP\left( E_\cI ; \, \gd_N=1\right)$.

\medskip

\begin{proposition}
\label{th:eta}
%%F
Assume that $\alpha=1/2$ and that \eqref{eq:Lassumption} holds
for some $\epsilon\in(0,1/2]$. For every $\eta>0$ and every $q>(2\epsilon)^{-1}$
we can choose $A>0$ such that if $\gb \le\gb_0$ and $h\le \Delta(\beta;q,A)$,
 for every $\mathcal I\subset\left\{1,\dots,m\right\}$ with $m\in\mathcal I$ we have
\begin{equation}
\label{eq:eta}
 \bbE\left[g_{\mathcal I}(\go)\hat Z_{\go}^{\mathcal I}\right]
 \,\le\, 
  \eta^{\vert\mathcal I\vert }P_{\mathcal I}.
\end{equation}
\end{proposition}

\medskip

 The following technical estimate controls $P_{\mathcal I}$ 
 (recall that $\cI=\{i_1, \ldots, i_{\vert \cI\vert}\}$).
\medskip

\begin{lemma} 
\label{th:P_cI}
Assume $\ga=1/2$. 
There exist $C_1=C_1(L(\cdot), k)$,
 $C_2=C_2(L(\cdot))$  and $k_0=k_0(L(\cdot))$ such that (with $i_0:=0$) 
\begin{equation}
\label{eq:P_cI}
P_\cI \, 
\le \,  C_1C_2^{\vert\cI\vert}\prod_{j=1}^{\vert\cI\vert}\frac{1}{(i_j-i_{j-1})^{7/5}},
\end{equation}
for $k \ge k_0$. 
\end{lemma}

\medskip

Note that in this statement $k$ is just a natural number, but we will apply
it with $k$ as in \eqref{eq:k} so that $k \ge k_0$
is just a requirement on $A$. Note also that the choice of $7/5$ is arbitrary
(any number in $(1, 3/2)$ would do: the constants $C_1$ and $C_2$
depend on such a number).

\smallskip

Let us now go back to \eqref{eq:fHolder}
and let us plug it into \eqref{eq:cg+fm} and use
Proposition~\ref{th:eta} and Lemma~\ref{th:P_cI} to get:
\begin{equation}
\bbE\left[ Z_{N , \go}^\gamma\right]
\, \le\, C_1^\gamma \sumtwo{\mathcal I\subset \left\{1,\dots,m\right\}}{m\in\mathcal I}
\prod_{j=1}^{\vert\cI\vert}
\left(
\frac{(2C_2 \eta)^\gamma}{(i_j-i_{j-1})^{7\gamma/5}}\right).
\end{equation}
But $7\gamma/5=6/5 >1$, so
we can choose
\begin{equation}
 \label{eq:fmm2-1}
\eta \, :=\, \frac1{3C_2 
\left(\sum_{i=1}^\infty i^{-6/5}
\right)^{7/6}
}, 
\end{equation}
and this implies that $n \mapsto (2C_2 \eta)^\gamma n^{-7\gamma/5}$
is a sub-probability, which directly entails that
 \begin{equation}
 \label{eq:fmm2}
\bbE\left[ Z_{N , \go}^\gamma\right]
\, \le\, C_1^\gamma,  \quad \quad \quad (\gamma=6/7)
\end{equation}
for every $N$, which implies, via \eqref{eq:fmm},
that $\tf(\gb, h)=0$ and we are done.

It is important to stress that $C_1$ may depend
 on $k$ (we need \eqref{eq:fmm2}  uniform in
 $N$, not in $k$), but $C_2$ does not ($C_2$ is just a function
 of $L(\cdot)$, that is a function of the chosen renewal), so that $\eta$ may actually be chosen
 a priori as in \eqref{eq:fmm2-1}: it is a small but fixed constant that depends only on
 the underlying renewal $\tau$. 

\medskip 

\begin{rem}\rm
\label{rem:kgD}
In this section we have actually hidden the role of 
$\gD(\gb; q,A)$ in the hypotheses of Proposition~\ref{th:eta},
which are the hypotheses of Theorem~\ref{th:main}. 
Let us therefore explain informally why we can prove a  critical point shift of  $\gD=1/k$.

The coarse graining procedure reduces proving delocalization to 
Proposition~\ref{th:eta}. As it is quite intuitive from \eqref{eq:decomp}--\eqref{eq:fHolder} and 
Figure~\ref{fig:cg},
one has to estimate  the expectation, with respect to
the $\bar g (\go)$-modified measure, of the partition function 
$Z_{d_j, f_j}$ (or, equivalently, $Z_{d_j, f_j}/\bP(f_j-d_j\in \tau)$)
in each visited block
(let us assume that $f_j-d_j$ is {\sl of the order of $k$},
because if it is much smaller than $k$ one can bound this
contribution in a much more elementary way). The Boltzmann factor in $Z_{d_j, f_j}$
is $\exp( \sum_{n=d_j+1}^{f_j} (\gb \go_n -\log \M (\gb) +h) \gd_n)$
which can be bounded (in an apparently very rough way)
by $\exp( \sum_{n=d_j+1}^{f_j} (\gb \go_n -\log \M (\gb) ) \gd_n)
\exp(hk)$, since $f_j-d_j \le k$. Therefore, if $h \le \Delta(\beta;q,A)\sim 1/k$ we can drop
the dependence on $h$ at the expense of the multiplicative factor $e$
that is innocuous because we can show that
the expectation (with respect to
the $\bar g (\go)$-modified measure) of 
$Z_{d_j, f_j}/\bP(f_j-d_j\in \tau)$ when $h=0$ can be made 
arbitrarily small by choosing $A$ sufficiently large. 
\end{rem}

\medskip

\begin{rem}\rm
\label{rem:lazy}
A last observation on the proof is about $\gb_0$.
It can be chosen arbitrarily, but for the sake of simplifying
the constants appearing in the proofs we choose 
 $\gb_0\in (0, \infty)$
such that
\begin{equation}
\label{eq:gb0}
\frac 12 \, \le \, 
\frac {\dd^2}{\dd \gb^2} \log \M( \gb)\, \le \, 2 \,  ,
\end{equation}
for $\gb \in [0, \gb_0]$.
Choosing $\gb_0$ arbitrarily just boils down to changing
the constants in the right-most and left-most terms in \eqref{eq:gb0}.
\end{rem}

\section{Coarse graining estimates}
\label{sec:cg-est}

We start by proving Lemma~\ref{th:P_cI}, namely
\eqref{eq:P_cI}. The proof is however more clear if instead
of working with the exponent $7/5$ we work with
$3/2-\xi$ ($\xi \in (0,1/2)$, in the end, plug in $\xi=1/10$).
\medskip

\noindent {\it Proof of 
Lemma~\ref{th:P_cI}}.
First of all, 
  in the product on the right--hand side of \eqref{eq:P_cI} one can clearly ignore the terms
  such that
  $i_j-i_{j-1}=1$. %which might just lead to a modification of the
                   %constant $C_2$
%(in the sequel, $C_2 \ge 1$). 
We then express $\cI$ in a more practical way by observing that
we can define, in a unique way, an integer $p\le l:=
\vert \cI\vert$ and increasing sequences of integers $\{a_j\}_{ j=1, \ldots, p} $, 
$\{b_j\}_{ j=1, \ldots,  p}$ with $b_p=m$, $a_j\ge b_{j-1}+2$ %%F
(for $j>1$) and $b_{j}\ge a_j$ such that
\begin{equation}
 \mathcal I= \bigcup_{j=1}^p [a_j,b_j]\cap \N.
\end{equation}
With this definition, it is sufficient  to show
\begin{equation}
 P_{\mathcal I}\le  C_1 C_2^l \frac1{a_1^{3/2-\xi}}\prod_{j=1}^{p-1}\frac{1}{(a_{j+1}-b_{j})^{3/2-\xi}}.
\end{equation}
We start then by writing
\begin{multline}\label{eq:upbound}
%%F
 P_{\mathcal I}
\le
\sumtwo{d_1\in B_{a_1}}{f_1\in B_{b_1}}\ldots\sumtwo{d_{p-1}\in B_{a_{p-1}}}{f_p\in B_{b_{p-1}}}\sum_{d_p\in B_{a_p}}K(d_1)\bP(f_1-d_1\in \tau)\ldots K(d_p-f_{p-1})\bP(N-d_p\in \tau),
\end{multline}
%%F
where the inequality comes from neglecting the constraint that
$\tau$ has to intersect $B_{a_j+1},\ldots$  $B_{b_j-1}$.
Note that the meaning of the $d$ and $f$ indexes is somewhat different
with respect to \eqref{eq:decomp1} and that in the above sum we always
have
\begin{equation}\begin{split}
 d_1&\in B_{a_1},\\
 (a_j-b_{j-1}-1)k&\le d_j-f_{j-1}\le (a_j-b_{j-1}+1)k,\\
 (b_j-a_{j}-1)k \vee 0&\le f_j-d_j\le (b_j-a_j+1)k.
 \end{split}\end{equation}
In particular, $f_j\ge d_j$ is guaranteed by the fact that 
$\bP(f_j-d_j\in\tau)=0$ otherwise.

 Observe now that for $k $ sufficiently large
\begin{equation}
 \sum_{x\in B_{a_1}}K(x)\, \le\,  
\begin{cases} 
1& \text{ if } a_1=1, \\
 3 \frac{L((a_1-1) k)}{k^{1/2}(a_1-1)^{3/2}}& \text{ if } a_1=2,3, \ldots, 
 \end{cases}
 \, \le \, c_1(k) \frac{ L(a_1k)}{k^{1/2} a_1^{3/2}},
\end{equation}
where $c_1(k):= \max ( 10, k^{1/2}/L(k))$.
 Moreover
there exists a constant $c_2$ depending on $L(\cdot)$ such that 
%%F 
for $j>1$
\begin{equation}\begin{split}
 \sum_{x=(a_j-b_{j-1}-1)k}^{(a_j-b_{j-1}+1)k}K(x)&\le c_2 \frac{L(k(a_j-b_{j-1}))}
 {k^{1/2}(a_j-b_{j-1})^{3/2}},\\
 \sum_{x=(b_j-a_{j}-1)k\vee 0}^{(b_j-a_j+1)k}\bP(x\in \tau)&\le c_2 \frac{k^{1/2}}{(b_j-a_j+1)^{1/2}L\left(k(b_j-a_j+1)\right)}.
\end{split}\end{equation}
The first inequality is obtained by making use of $a_j\ge b_{j-1}+2$.
Neglecting the last term which is smaller than one, we can bound the right--hand side of \eqref{eq:upbound} and get
\begin{equation}
\label{eq:truc}
 P_{\mathcal I}\le c_1(k)c_2^{2p}  \frac{ L(a_1k)}{k^{1/2} a_1^{3/2}}\prod_{j=1}^{p-1}
 \left( \frac{L(k(a_{j+1}-b_{j}))}{(a_{j+1}-b_{j})^{3/2}} \right)
 \left(  \frac{1}{(b_j-a_j+1)^{1/2}L\left(k(b_j-a_j+1)\right)}\right)
.
\end{equation}
Notice now that since $L(\cdot)$ grows slower than any power,
 $ \sup_{a_1}{ L(a_1k)}/({k^{1/2} a_1^{\xi}})$ is $o(1)$ for $k$ large.
 To control the other terms 
 we use the {\sl Potter bound} \cite[Th.~1.5.6]{cf:RegVar}: given
 a slowly varying function $L(\cdot)$ which is locally bounded away
 from zero and infinity (which we may assume in our set up without loss of generality),
 for every $a>0$ there exists $c_a>0$ such that for every $x,y >0$
 \begin{equation}
\label{eq:Potter}
\frac{L(x)}{L(y)} \, \le \,  c_a \max \left(\frac{x}{y}, \frac{y}{x} \right)^a.
 \end{equation}
This bound implies 
that for large enough $k$
\begin{equation}
\label{eq:2ineqs}
 \sup_{x\ge 1}\frac{L(k)}{\sqrt{x}L(kx)}\, \le\,  2 \ \text{ and } \
 \sup_{x\ge 1} \frac{L(kx)}{L(k)x^{\xi}}\,\le \,2.
\end{equation}
In fact consider the second bound (the argument for the first one is identical): by choosing $a=\xi/2$ we have
 $L(kx)/(L(k)x^{\xi})\le c_{\xi/2} x^{-\xi/2} \le 2$ and the second
 inequality holds for $x$ larger than a suitable constant $C_\xi$.
 For $x(\ge 1)$ smaller than $C_\xi$ instead it suffices
 to choose $k$ sufficiently large so that $L(kx)/L(k) \le 2$
 for every $x \in [1, C_\xi]$. 
Using the two bounds \eqref{eq:2ineqs} 
in \eqref{eq:truc} we complete the proof.
\qed

\medskip
The proof of 
Proposition~\ref{th:eta} depends on the following lemma that will be proven
in the next section. 
\medskip

\begin{lemma}
\label{th:df}
Set $h=0$, fix $q\in \N$, $q> (2\epsilon)^{-1}$ as in Theorem 
\ref{th:main},  and 
recall the definition of $k=k(\gb; q,A)$ \eqref{eq:k}.
For every $\gep$ and $\gd>0$ there exists $A_0>0$ such that
for $A\ge A_0$
\begin{equation}
 \bbE\left[\bar g(\go) z_d Z_{d,f}\right]\, \le\,
  \delta\,  \bP(f-d\in \tau),
 \label{eq:df}
\end{equation}
for every  $d$ and $f$ such that $0\le d\le d+ \gep k\le f\le k$
and $\gb \le \gb_0$. 
\end{lemma}
\medskip

\begin{proof}[Proof of Proposition~\ref{th:eta}]
Recalling \eqref{eq:decomp1} and the notations for the set $\cI$ in there, we have
\begin{multline}
\label{eq:longlong}
  \bbE\left[g_{\mathcal I}(\go)\hat Z_{\go}^{\mathcal I}\right]\,=\\
  \sumtwo{d_1, f_1 \in B_{i_1}}{d_1\le f_1}
  \sumtwo{d_2, f_2\in B_{i_2}}{d_2\le f_2}\ldots \sum_{d_l\in B_{i_l}}
  K(d_1)\bbE\left[\bar g(\go) z_{d_1-k(i_1-1)}Z_{d_1-k(i_1-1),f_1-k(i_1-1)}\right] K(d_2-f_1)
  \ldots \\ \quad \quad \quad \quad \quad \quad \quad \quad \quad \quad \quad \quad 
  K(d_l-f_{l-1})\bbE\left[\bar g(\go) z_{d_l-k(m-1)} Z_{d_l-k(m-1),k}\right]\\
  \le e^l \sumtwo{d_1, f_1 \in B_{i_1}}{d_1\le f_1}
\sumtwo{d_2, f_2\in B_{i_2}}{d_2\le f_2}\ldots \sum_{d_l\in B_{i_l}}
  K(d_1)(\gd+\ind_{\{f_1-d_1\le \gep k\}})\bP(f_1-d_1\in\tau)
  K(d_2-f_1)\ldots \\
  K(d_l-f_{l-1})(\gd+\ind_{\{N-d_l\le \gep
    k\}})\bP(N-d_l\in \tau),
\end{multline}
where the factor $e^l$ in the last expression 
comes from bounding the contribution due to $h$ (recall that $h k \le 1$).
We now consider $B_{i_j}$ as the union of two sub-blocks
\begin{equation}\begin{split}
 B_{i_j}^{(1)}&:=\left\{(i_j-1)k,\dots,(i_j-1)k+\lfloor k/2 \rfloor \right\},\\
 B_{i_j}^{(2)}&:=\left\{(i_j-1)k+\lceil k/2 \rceil,\dots, i_jk \right\}.
\end{split}\end{equation}
If $d_j\in B_{i_j}^{(1)}$ then if $\gep$ is sufficiently small
($\gep \le 1/10$ suffices) we have that for $k$ sufficiently large
({\sl i.e.} $k \ge k_0(L(\cdot), \gep)$)
\begin{equation}
\sum_{f=d_j}^{d_j+\gep k}
\bP(f-d_j\in \tau )K(d_{j+1}-f)
\le 4 \left(\sum_{x=1}^{k\gep} \bP(x\in \tau)\right)
K(k(i_{j+1}-i_j)).
\end{equation}
This can be compared to
\begin{equation}
\sum_{f=d_j}^{ki_j}\bP(f-d_j\in \tau )K(d_{j+1}-f)
\, \ge\,  \frac{1}{3} \left(\sum_{x=1}^{\lfloor k/4\rfloor}\bP(x\in \tau )\right) 
K(k(i_{j+1}-i_j)),
\end{equation}
that holds once again for $k$ large. By using that
 $\sum_{x=1}^n \bP(x \in \tau)$
behaves for $n $ large like $\sqrt{n}$ times a slowly varying function 
({\sl cf.} \eqref{eq:Doney1/2}) we therefore see that
 given $\delta>0$ we can  find $\gep$ such that for any $d_j\in B_{i_j}^{(1)}$ we have
\begin{equation}
\sum_{f=d_j}^{d_j+\gep k}\bP(f-d_j\in \tau )K(d_{j+1}-f)\le \gd \sum_{f=d_j}^{ki_j}\bP(f-d_j\in \tau ) K(d_{j+1}-f).
\end{equation}
Using the same argument in the opposite way one finds that if $f_j\in B_{i_j}^{(2)}$
\begin{equation}
 \sum_{d=f_j-\gep k}^{f_j}K(d-f_{j-1})\bP(f_j-d\in \tau )\, \le\,
  \gd \sum_{d=k(i_j-1)}^{f_j}K(d-f_{j-1})\bP(f_j-d\in \tau ).
\end{equation}
Since either $d_j\in B_{i_j}^{(1)}$ or $f_j\in B_{i_j}^{(2)}$, we conclude that
\begin{multline}
 \sumtwo{d_j, f_j \in B_{i_j}}{d_j\le f_j}\ind_{\{f_j-d_j\le k\gep \}} K(d_j-f_{j-1})\bP(f_j-d_j\in \tau)K(d_{j+1}-f_j)\\
\le \gd \sumtwo{d_j, f_j \in B_{i_k}}{d_j\le f_j}  K(d_j-f_{j-1})\bP(f_j-d_j\in \tau)K(d_{j+1}-f_j).
\end{multline}
The analog estimate can be obtained for the sum over $d_l$
in \eqref{eq:longlong} (rather, it is  easier). Using this inequality $j=1\dots l$ we get our result for $\eta=2e \gd$. 

\end{proof}

\section{The $q$-body potential estimates (proof of Lemma~\ref{th:df})}

In what follows $X=X_1$ and we fix $\gd\in (0,1)$. The positive (small) number
$\gep$ is fixed too, as well as $q>(2\epsilon)^{-1}$, where 
$\epsilon$ is the same which appears in the statement 
of Theorem \ref{th:main}.

%The value of $\gb_0$ is chosen along the proof, by a finite number of (smallness) conditions. 

\medskip

\noindent
{\it Proof of Lemma~\ref{th:df}.}
We start by observing that, since $h=0$,
\begin{equation}
\bbE\left[ \bar g(\go) z_d Z_{d,f} \right]\, =\, 
\bE_{d,f}\left[
\bbE \left[ \bar g(\go) \exp \left( \sum_{n=d}^f
(\gb \go_n - \log \M (\gb) ) \gd_n \right) \right] \right] \bP(f-d \in \tau),
\end{equation}
where $\bP_{d,f}$ is the law of  $\tau \cap [d,f]$, conditioned to
$f,d \in \tau $. Given the random set (or renewal trajectory) $\tau$ we  introduce the probability measure
\begin{equation}\label{eq:ptau}
\hat \bbP_{\tau}(\dd\go)\, :=\,  \exp \left( \sum_{n=d}^f
(\gb \go_n - \log \M (\gb) ) \gd_n \right) \bbP(\dd \go).
\end{equation}
Note that $\go$, under $\hat \bbP_\tau$, is still a sequence
of independent random variables, but they are no longer identically distributed.
We will use that, for $d\le n\le f$,
%%F
\begin{equation}
\label{eq:2beused}
\hat \bbE_\tau \go_n \, =\, \mbeta  \gd_n \stackrel{\gb \searrow 0}\sim \gb \gd_n
\ \text{(so that }\gb/2 \le \mbeta \le 2 \gb \, \text{)} 
\ \  \text{ and } \  \
\text{var}_{\hat \bbP_\tau} \left( \go_n \right) \, \le \, 2, 
\end{equation}
where the inequalities hold for $\gb\le \gb_0$ (recall \eqref{eq:gb0}) and all relations hold uniformly in the 
renewal trajectory
$\tau$. 
%%F
On the other hand, for $n\notin \{d,\ldots,f\}$ the $\go_n$'s are IID exactly as under
$\bbP$.
We have:
\begin{multline}
\label{eq:3lhat}
\frac{
\bbE\left[ \bar g(\go) z_d Z_{d,f} \right]}
{\bP(f-d \in \tau)}
\, =\, \bE_{d,f} \hat \bbE_{\tau} \left[ \bar g(\go) \right]\, =
\\ 
\exp(-K) 
 \bE_{d,f} \hat \bbP_{\tau} \left[X \ge \exp(K^2) \right]
 +  \bE_{d,f} \hat \bbP_{\tau} \left[ X< \exp(K^2) \right]
 \, \le \\
  \exp(-K) + \bE_{d,f} \hat \bbP_{\tau} \left[ X< \exp(K^2) \right]\, \le \, 
 \frac {\gd }3 +  
  \bE_{d,f} \hat \bbP_{\tau} \left[ X< \exp(K^2) \right],
\end{multline}
where in the last step we have chosen $K $ such that 
$\exp(-K ) \le \gd /3$. 
We are now going to use the following lemma:

\medskip

\begin{lemma}
\label{th:fromCE}
If $d$ and $f$ are chosen such that 
$f-d \ge \gep k$ and 
 $X(=X_1)$ is defined as in \eqref{eq:X}, that is
$X= \sum_{{\ui}\in B_1^q} V_k({\ui})\go_{\ui}$, 
 we have  that for every $\zeta>0$ we can find $a>0$ and $A_0$ such that
\begin{equation}
\label{eq:fromCE}
 \bP_{d,f} \left( \hat \bbE_\tau X \, >\, aA^{(q-1)/2}\right)\, \ge \, 1-\zeta,
\end{equation}
for $\gb\le \gb_0$ and $A\ge A_0$.
\end{lemma}
\medskip

We apply this lemma by setting $\zeta=\gd/3$ (so $a$ is fixed once
$\gd$ is chosen)
so that,
if we choose $K$ such that $2\exp(K^2)=aA^{(q-1)/2}$
(note that, by choosing $A$ large we make $K$ large
and we automatically satisfy the previous requirements on $K$),
we have 
$\bP_{d,f} \left( \hat \bbE_\tau X < 2
\exp(K^2)\right) \le \gd/3  $, so that, in view of \eqref{eq:3lhat},
we obtain
\begin{equation}
\label{eq:2/3and}
\begin{split}
\frac{
\bbE\left[ \bar g(\go) z_d Z_{d,f} \right]}
{\bP(f-d \in \tau)}
\, &\le \, \frac {2\gd}3\, +\,  \bE_{d,f} \hat \bbP_{\tau} \left[ X-\hat \bbE_\tau X \le  -\exp(K^2) \right]
\\
%%F
&\le \, \frac {2\gd}3\, +\,  \frac{4}{a^2 A^{q-1}}
\bE_{d,f} \hat \bbE_{\tau} \left[ \left(X-\hat \bbE_\tau X\right)^2 \right]\, .
\end{split}
\end{equation}

The conclusion now follows 
as soon as we can show that the second moment appearing  in the last term of \eqref{eq:2/3and}
is $o(A^{q-1})$ for $A$ large. But this is precisely 
what is granted by the 
next lemma:

\medskip

\begin{lemma}
\label{th:square}
There exist $A_0>0$  such that
\begin{equation}
\label{eq:square}
 \bE_{d,f} \hat \bbE_{\tau} \left[ \left(
X-\hat \bbE_\tau X \right)^2 \right] \, \le \, A^{(q-1)^2/q} ,
\end{equation}
 for every $\gb \le \gb_0$ and every $A\ge A_0$. 
\end{lemma}

\medskip

\noindent
{\it Proof of Lemma~\ref{th:square}.}
We start by introducing the notation 
%%F
$\hat \go_{n} :=
\go_{n} - \mbeta \gd_{n}{\bf 1}_{\{d\le n\le f\}}$ and by observing that
%%F
\begin{equation}
\begin{split}
\hat \bbE_{\tau} \left[ \left(
X-\hat \bbE_\tau X \right)^2 \right]
\, &= \, 
\hat \bbE_{\tau} \left[ \left( \sum_{\ui \in B_1^q} V_k({\ui}) 
\prod_{a=1}^q \left( \hat \go_{i_a} + \mbeta \gd_{i_a} 
{\bf 1}_{\{d\le i_a\le f\}}
\right)
- \mbeta ^q \sum_{\ui \in \{d,\ldots,f\}^q} V_k({\ui}) \gd_{\ui}
\right) ^2\right]
\\
&\le C(q)\, \hat \bbE_{\tau} \left[\left(
\sum_{\ell=0}^{q-1} \mbeta ^{\ell} \sum_{\ui \in B_1^{q-\ell}}
\sum_{\uj \in \{d,\ldots,f\}^\ell}  V_k({\ui \,  \uj}) \hat \go_{\ui} \gd_{\uj} \right)^2\right]
\\
& \, \le \, C(q)
\sum_{\ell=0}^{q-1} \mbeta ^{2\ell}
\sum_{\ui \in B_1^{q-\ell}}
\sum_{\uj,\um \in \{d,\ldots,f\}^\ell} %\sum_{\um \in B_1^\ell}
 V_k({\ui \,  \uj}) V_k({\ui \,  \um}) \gd_{\uj} \gd_{\um} 
\, ,
\end{split}
\end{equation}
where $\ui \, \uj \in B_1^q$ is the concatenation of $\ui$ and $ \uj$
and in the last step we have first used the Cauchy-Schwarz inequality,
 %%F
the fact that the $\hat \go$ variables are independent and centered and
\eqref{eq:2beused}.  
%%F
\begin{rem}
  \rm
\label{rem:nonumero}
Here and in the following, we adopt the convention that
$C(a,b,\ldots)$ is a positive constant (which depends on the
parameters $a,b,\ldots$), whose numerical value may change from line
to line.
\end{rem}
Therefore
\begin{equation}
\label{eq:bfl1}
\bE_{d,f}
\hat \bbE_{\tau} \left[ \left(
X-\hat \bbE_\tau X \right)^2 \right]
\, \le \, C(q)
\sum_{\ell=0}^{q-1} \mbeta ^{2\ell}
\sum_{\ui \in B_1^{q-\ell}}
\sum_{\uj, \um \in \{d, \ldots , f\}^\ell}
 V_k({\ui \,  \uj}) V_k({\ui \,  \um}) \bE_{d,f} \left[\gd_{\uj} \gd_{\um} \right].
\end{equation}
Let us point out immediately that we know how to deal with the $\ell=0$ case:
it is simply $ C(q) \sum_{\ui \in B_1^q} V_k({\ui})^2$  and it is therefore bounded by 
$ 2C(q)$ ({\sl cf.} \eqref{eq:cV}).
By using the notation and the bounds in Remarks \ref{rem:L} and \ref{rem:V},
together with the renewal property, we readily see that 
\begin{equation}
\label{eq:useR}
\bE_{d,f} \left[\gd_{\uj} \gd_{\um} \right] \, \le \, 
\frac{\const _L ^{-(2\ell+1)}}{\bP( f-d \in \tau)}
\prod_{a=1}^{ 2\ell+1}
R_{\frac12} \left( r_a-r_{a-1}\right)\, \le \, 
\frac{\const _L ^{-(2\ell+1)}}{R_{\frac12} \left(f-d \right)}
\prod_{a=1}^{ 2\ell+1}
R_{\frac12} \left( r_a-r_{a-1}\right),
\end{equation}
for $\uj , \um\in \{d, \ldots, f\}^\ell$, $r= \sort(\uj \, \um)$,
$r_0:=d$ and $r_{2\ell+1}:=f$.  A notational simplification may be
therefore achieved by exploiting further Remark~\ref{rem:V}, namely by
using \eqref{eq:U}, so that \eqref{eq:useR} becomes
\begin{equation}
\label{eq:useR.1}
\begin{split}
\bE_{d,f} \left[\gd_{\uj} \gd_{\um} \right] \, &\le \, 
\const _L ^{-(2\ell+1)}
R_{\frac12} \left(f-d \right)^{-1} R_{\frac12}(\min(\uj\, \um) -d) 
U(\uj\, \um) R_{\frac12}(f-\max(\uj\, \um) )
\\
&= \, 
\const _L ^{-(2\ell+1)}
R_{\frac12} \left(f-d \right)^{-1} 
U(d\, \uj\, \um \, f) \, 
.
\end{split}
\end{equation}
By inserting \eqref{eq:useR.1} and \eqref{eq:useR2} into 
\eqref{eq:bfl1} we get to 
%%F
\begin{equation}
\label{eq:bfl2}
\begin{split}
%%F
&\bE_{d,f} 
\hat \bbE_{\tau} \left[ \left(
X-\hat \bbE_\tau X \right)^2 \right]%\, -\, q 2^{q+1} 
\\ 
& \le  C\left(1+\frac1
%{\const _L ^{-(2q-1)}}
{k \tilde L(k)^{q-1}R_{\frac 12}(f-d) } 
\sum_{\ell=1}^{q-1} \mbeta ^{2\ell}
\sum_{\ui \in B_1^{q-\ell}}
 \sum_{\uj, \um \in \{d, \ldots , f\}^\ell}
U(\ui\, \uj) U(\ui\, \um) U(d\, \uj\, \um \, f)\right)
\\
& \le C \left(1+
 \frac1
%{\const _L ^{-(2q-1)}}
{k \tilde L(k)^{q-1} R_{\frac 12}(f-d)} 
\sum_{\ell=1}^{q-1} \mbeta ^{2\ell} %(q-\ell)! (\ell!)^2
\sum_{\ui \in \sort(B_1^{q-\ell})}
\sum_{\uj, \um \in \sort(\{d, \ldots , f\}^\ell)}
U(\ui\, \uj) U(\ui\, \um) U(d\, \uj\, \um\, f)\right),
%\\
%&\le  
%\frac{q 2^q (q!)^2 \const _L ^{-(2q-1)}}{k \tilde L(k)^{q-1} R_{\frac 12}(f-d)} 
%\sum_{\ell=1}^{q-1} \mbeta ^{2\ell} 
%\sum_{\ui \in \sort(B_1^{q-\ell})}
%\sum_{\uj, \um \in \sort(\{d, \ldots , f\}^\ell)}
%U(\ui\, \uj) U(\ui\, \um) U(d\, \uj\, \um \,f),
\end{split}
\end{equation}
where of course $\sort(\{1, \ldots, a\}^n)= \{ \ui \in \{1, \ldots,
a\}^n: \, i_1 \le i_2 \le \ldots \le i_n\}$ and $C=C(q,L(\cdot))$,
with the convention of Remark \ref{rem:nonumero}.

\smallskip

The rest of the proof is devoted to bounding
\begin{equation}
\label{eq:Tell}
T_{q,\ell} \, := \, 
\sum_{\ui \in \sort(B_1^{q-\ell})}
\sum_{\uj, \um \in \sort(\{d, \ldots , f\}^\ell)}
U(\ui\, \uj) U(\ui\, \um) U(d\, \uj\, \um \,f).
\end{equation}
This is relatively heavy, because, while $\ui$, $\uj$ and $\um$ are
ordered, $\ui\, \uj$, $\ui\, \um$ and $\uj\, \um$ are not.  We have
therefore to estimate the contributions given by every mutual
arrangement of $\ui$, $\uj$ and $\um$.  This will be done in a
systematic way with the help of a {\sl diagram representation} (the
diagrams will correspond to groups of configurations $\ui$, $\uj$ and
$\um$ that have the same {\sl mutual order}).  \smallskip

Fix $q$ and $\ell$ and choose $\ui\in \sort(\{1, \ldots, k\}^{q-\ell})$ and  $\uj, \um
\in \sort(\in \{d, \ldots, f\}^{\ell})$.
%%F
The construction of the diagram of  $\ui$, $\uj$ and $\um$ is done in steps:

\smallskip

\begin{enumerate}
\item
Mark with $\Box$'s
on the horizontal axis (the dotted line in Figure~\ref{fig:diagram-ex})
  the positions $i_1 \le i_2\le \dots\le i_{q-\ell}$.
Do the  same for $\uj$ (using $\circ$) and $\um$ (using $\bullet$).
As explained in Remark~\ref{rem:superpos} below, we may and do assume 
that symbols do not
 sit on the same position (this amounts to assuming
 strict inequality between all  indexes).
\item 
%%F
Consider the set of  $\Box$'s and $\circ$'s, and connect
all nearest neighbors  with a line (the line may be
straight or curved for the sake of visual clarity).   
\item Do the same for the set of  $\Box$'s and $\bullet$'s.
\item Do the same for the set of  $\circ$'s and $\bullet$'s.
\item Consider   the set of $\circ$'s and $\bullet$'s and connect the element that
is closest to $d$ with $d$. Do the analogous action  with the element which
is closest to $f$. The point $d$ is always to the left of $\circ$'s and $\bullet$'s
and the point $f$ is always to the right.
\end{enumerate}

\smallskip

We have now a graph with vertex set $\{d,f, \ui, \uj, \um\}$.
Vertices have a type ($\Box$, $\circ$ and $\bullet$): $d$ and $f$ have
their own type too, graphically this type is $\vert$.  We actually
consider the {\sl richer} graph with vertex set given by the points
and the type of the point.  The edges are the ones built with the
above procedure; note that there may be double edges: we keep them and
call them {\sl twin} edges.  Two indexes configurations are equivalent
if they can be transformed into each other by translating the indexes
without allowing them cross (and, of course, keeping their type; the
vertices $d$ and $f$ are fixed).  This leads to equivalence classes
and a class is denoted by $\cG$: we split the sum in \eqref{eq:Tell}
according to these classes, that is $T_{q,\ell}= \sum_{\cG} T_{q,\ell,
  \cG}$. The bound we are going to find is rather rough: we are going
in fact to bound $\max _{\cG} T_{q,\ell, \cG}$.

\smallskip

\begin{rem}
\label{rem:superpos}
\rm
We have built equivalent classes of non-superposing points only.
However in estimating $T_{q,\ell, \cG}$ we will allow the index summations to include
coinciding indexes so  in the end we include (and over-estimate) the contributions
of all the configurations of indexes. 
\end{rem}

\smallskip

\begin{figure}[hlt]
\begin{center}
\leavevmode
\epsfxsize =14.5 cm
\psfragscanon
\psfrag{0}[c][l]{\small $0$}
\psfrag{k}[c][l]{\small $k$}
\psfrag{d}[c][l]{\small $d$}
\psfrag{f}[c][l]{\small $f$}
\psfrag{1}[c][l]{\small Start}
\psfrag{2}[c][l]{\small Trim step $1$}
\psfrag{3}[c][l]{\small Trim step $2$}
\psfrag{4}[c][l]{\small Trim step $3$}
\psfrag{5}[c][l]{\small Trim step $4$}
\psfrag{i1}[c][l]{\small $i_1$}
\psfrag{i2}[c][l]{\small $i_2$}
\psfrag{i3}[c][l]{\small $i_3$}
\psfrag{i4}[c][l]{\small $i_4$}
\psfrag{i5}[c][l]{\small $i_5$}
\psfrag{j1}[c][l]{\small $j_1$}
\psfrag{m1}[c][l]{\small $m_1$}
\psfrag{j2}[c][l]{\small $j_2$}
\psfrag{m2}[c][l]{\small $m_2$}
\epsfbox{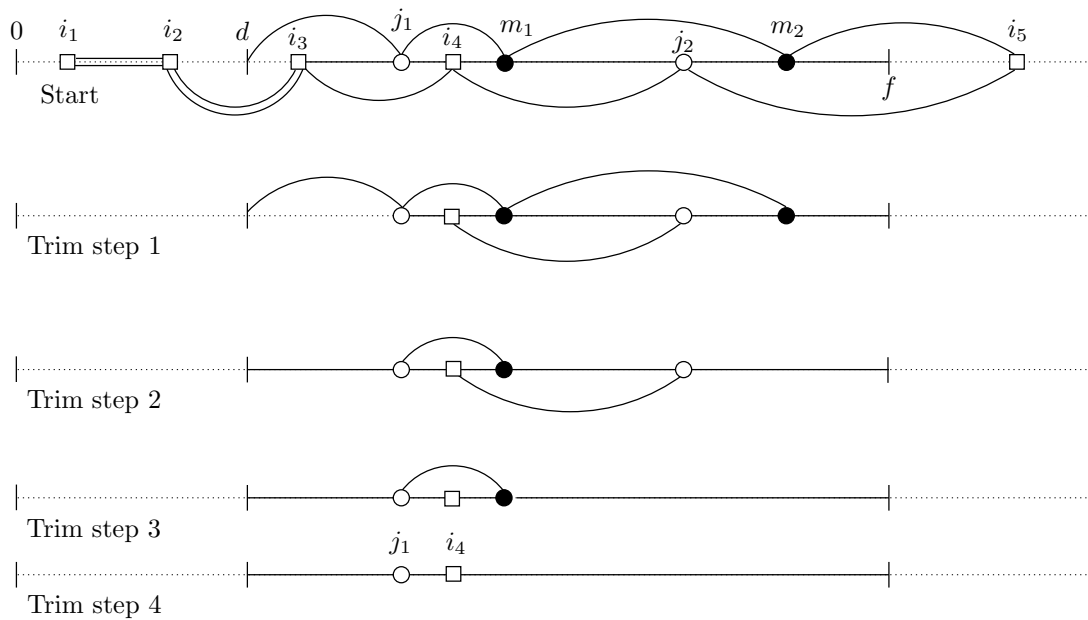}
\caption{\label{fig:diagram-ex}A diagram arising for $q=7$ and $\ell=2$ and the successive trimming procedure explained in the text}
\end{center}
\end{figure}
 
 \medskip
 
 In order to estimate $T_{q,\ell, \cG}$ we proceed to a graph trimming
 procedure that will be then matched to successive estimates on 
  $T_{q,\ell, \cG}$.
 
 The trimming procedure is the following:
 \smallskip

\begin{enumerate}
\item If there are $\Box$ vertices that are left of leftmost element of the 
set of $\circ$ and $\bullet$ vertices (we may call these $\Box$
vertices {\sl external vertices}), we erase them and we trim
the edges linking them. Note that if we do this procedure left to right,
we  erase one vertex and two edges at a time: at each step 
we trim a couple of twin edges, except at the last step in
which the edges are not twin.
We do the same with the $\Box$ vertices that are right 
of the rightmost element of the 
set of $\circ$ and $\bullet$ vertices (if any, of course). The trimming procedure goes 
this time right to left. We call {\sl internal} the vertices that are left. 
\item Now we start (say) right and we erase the rightmost internal
  vertex (in this first step is necessarily a $\circ$ or a $\bullet$,
  later it may be a $\Box$; we do not touch $d$ and $f$). Note that it
  has two edges (linking to vertices on the left) and one edge linking
  it with $f$: we trim these three edges and we add an edge linking
  the rightmost vertex (it can have any type among $\Box$, $\circ$ and
  $\bullet$) that is still present to $f$ with an edge.
\item We repeat step (2) till one is left with only four  vertices (among them, 
only  one may be a $\Box$) and three edges. {\sl Trim step 4} in Figure~\ref{fig:diagram-ex}
is a possible fully trimmed configuration.
\end{enumerate}

\medskip

Let us now explain  the link between the trimming procedure and quantitative estimates
on $T_{q, \ell, \cG}$.
Also this is done by steps corresponding precisely to the three steps
of the trimming procedure:

\medskip

\begin{enumerate}
\item Consider the external $\Box$ vertices connected to the rest of the graph by twin edges, if any.
 We start by the leftmost (if there is at least one on the left: the procedure from the right is absolutely analogous)
and notice that we can sum over the index, that is $i_1$, and use that, thanks to \eqref{eq:Lb} (recall
\eqref{eq:Ltilde} and \eqref{eq:R12}), there exists $C_L$ such that for $0<n \le k$
\begin{equation}
\label{eq:basicbound}
\sum_{i=0}^n (R_{\frac 12}(n-i))^2 \, \le \, C_L \tilde L(k).
\end{equation}
We are of course over-estimating the real sums that are, in most cases,
restricted to small portions of $B_1$. 
This estimate allows {\sl trimming} $T_{q, \ell, \cG}$ in the sense that
it gives the bound $T_{q, \ell, \cG} \le C_L^r \tilde L(k)^r T_{q-r, \ell, \cG^\prime}$, 
where $r$ is the number of twin edges and $\cG^\prime$ is the graph, with $q-r+\ell$ vertices
that is left after this procedure. This step can be repeated also for the last 
external $\Box$'s (there are at most two, one on the left and one on the right). In these cases
we simply use that $R_{\frac 12}(\cdot)$ is decreasing so that if 
$0\le n \le n^\prime $
\begin{equation}
\label{eq:basicbound2}
\sum_{i=0}^n R_{\frac 12}(n-i)R_{\frac 12}(n^\prime -i)  \, \le \, \sum_{i=0}^n (R_{\frac 12}(n-i))^2,
\end{equation}
and then \eqref{eq:basicbound} applies. So this extra trimming yields again $C_L \tilde L(k)$
to the power of half the number of edges trimmed, that is, to the power of the number of the external vertices.
\item We are left with the internal vertices and we start erasing 
the vertex (it is necessarily $\circ$ or $\bullet$ at this stage)
which is most on the right. So we sum over its index and use the bound: there exists a constant $C_L$ such that
for $(0\le )d\le n^\prime \le n \le f(\le k)$
we have 
\begin{multline}
\label{eq:ffR}
\sum_{j=n}^f R_{\frac 12}(j-n)R_{\frac 12}(j-n^\prime)R_{\frac 12}(f-j)\, 
\le \, \sum_{j=0}^{f-n} R_{\frac 12}(j)^2 R_{\frac 12}((f-n)-j)
\\ 
\le \, C_L \tilde L (f-n) R_{\frac 12}(f-n) \, \le \, C_L \tilde L (k) R_{\frac 12}(f-n),
\end{multline}
where in the first inequality we have used the monotonicity of $R_{\frac 12}(\cdot)$,
in the second we have explicitly estimated the sum by using standard results
on regularly varying function and \eqref{eq:Ltilde-prop}. The last inequality is just the monotonicity
of $\tilde L(\cdot)$. This means that this trimming step brings once  again a multiplicative factor  
$C_L \tilde L (k)$: of course this time we have trimmed three edges, but we have also the extra 
factor $R_{\frac 12}(f-n)$ which is precisely the contribution of a longer edge that we rebuild
(see Figure~\ref{fig:diagram-gen}).
\item Keep repeating the previous step (the type of the vertices is not really important), 
trimming each time three edges, but rebuilding one too (so, in total, minus two edges), till
the graph with four vertices and three edges.
\end{enumerate}

\medskip

\begin{figure}[hlt]
\begin{center}
\leavevmode
\epsfxsize =12.5 cm
\psfragscanon
\psfrag{f}[c][l]{\small $f$}
\psfrag{a1}[c][l]{\small $a_1$}
\psfrag{a2}[c][l]{\small $a_2$}
\psfrag{a3}[c][l]{\small $a_3$}
\epsfbox{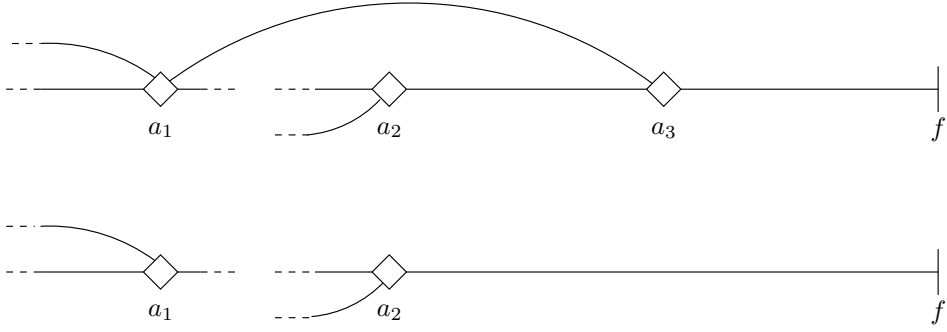}
\caption{\label{fig:diagram-gen} The second step of the trimming
  procedure corresponding to the estimate \eqref{eq:ffR}.  The symbol
  $\diamond$ may represent $\Box$, $\circ$ and $\bullet$: the choice
  is not fully arbitrary, in the sense that for example before
  starting the trimming procedure there is no edge between $f$ (or
  $d$) and a $\Box$.  However the estimate is independent of the type
  of symbols.}
\end{center}
\end{figure}

In order to evaluate the contribution of all the trimming procedure we
just need to count the number of vertices that we have erased: $q+\ell
-2$. We are now left with the contribution given by the last diagram
(four points, three edges: see for example
{\sl trim step 4} in Figure~\ref{fig:diagram-ex}), times of course $(C_L \tilde L(k))^{q+\ell
  -2}$: we bound  the last diagram using
\begin{equation}
\sum_{i=d}^f\sum_{j=i}^f R_{\frac 12}(i-d) R_{\frac 12}(j-i) R_{\frac 12}(f-j)
\, \le \, C_L \frac{\sqrt{f-d}}{L(f-d)^3}
\, \le \, C_{L, \gep} \frac{k\, R_{\frac 12}(f-d )}{L(k)^2}.
\end{equation}
where $C_L$ is once again a constant that depends only on $L(\cdot)$,
while in the last step we have used $k\ge f-d \ge \gep k$ and
\eqref{eq:Doney-bound}.  Going back to \eqref{eq:bfl2} we see that
there exists $C=C(\gep,q,L(\cdot))$ such that (with the convention of
Remark \ref{rem:nonumero})
%%F
\begin{multline}
\label{eq:varbfi}
\bE_{d,f} \hat \bbE_{\tau} \left[ \left( X-\hat \bbE_\tau X \right)^2
\right]\, \le \\ C\left(1 +\, \, \max_{\ell=1,2, \ldots, q-1}
  \frac{1}{k \tilde L(k)^{q-1} R_{\frac 12}(f-d)} \frac{\tilde
    L(k)^{q+\ell-2}k\, R_{\frac 12}(f-d
    )}{L(k)^2}\mbeta^{2\ell}\right)
\\
=\, C\left(1 + \, \max_{\ell=1,2, \ldots, q-1} \frac{\tilde L
    (k)^{\ell-1}}{L(k)^2} \mbeta^{2\ell}\right)\, \le \, C\left(1 \, +
  \, \max_{\ell=1,2, \ldots, q-1} \frac{\tilde L
    (k)^{\ell-1}}{L(k)^2} \gb^{2\ell}\,\right),
\end{multline}
where in the last line we have used $\mbeta \le 2 \gb$, for $\gb \le \gb_0$
({\sl cf.} \eqref{eq:2beused}). We now recall \eqref{eq:k} that guarantees that
\begin{equation}
\label{eq:LtL}
\frac{\tilde L (k-1)}{L(k-1) ^{2/(q-1)}} \gb ^{2q/(q-1)} \, < \, A
\ \ \text{ so that } \ \ 
\frac{\tilde L (k)}{L(k) ^{2/(q-1)}} \gb ^{2q/(q-1)} \, \le  \, 2A,
\end{equation}
where the second inequality is a consequence of the slowly varying character
of $L(\cdot) $ and $\tilde L(\cdot)$ and it holds for $k$ sufficiently large.
But this implies
\begin{equation}
\label{eq:llb}
 \frac{\tilde L (k)^{\ell-1}}{L(k)^2}
 \gb^{2\ell}\, \le \, (2A)^{(q-1)\ell/q} \, \left( \tilde L(k) L(k)^2\right)^{-1+(\ell/q)}
 \, , 
\end{equation}  
so that, by \eqref{eq:Ltilde-prop}, 
%%F
by choosing $A$ large we can make  
the quantity in \eqref{eq:llb} arbitrarily small (recall that $\ell=1, \ldots, q-1$),
so that going back to \eqref{eq:varbfi}, we see that
\begin{multline}
\bE_{d,f} 
\hat \bbE_{\tau} \left[ \left(
X-\hat \bbE_\tau X \right)^2 \right]\, \le \\
C(\gep,q, L(\cdot))\left(1 + A^{(q-1)^2/q}  
 \max_{\ell=1, \ldots, q-1}\left( \tilde L(k) L(k)^2\right)^{-1+(\ell/q)}\right) \,
\le \, A^{(q-1)^2/q} \, , 
\end{multline}
where in the last step we have used that, by \eqref{eq:Ltilde-prop},
the maximum in the intermediate term can be made arbitrarily small,
by choosing $k$ large (that is, $A$ larger than a constant depending 
on $\gep$, $q$ and $L(\cdot)$).
 This completes the proof of Lemma~\ref{th:square}.
\qed

\section{Some probability estimates (Proof of Lemma~\ref{th:fromCE})}

The proof is done in four steps.

\medskip

\noindent
{\it Step 1: reduction to an asymptotic estimate on a constrained renewal.}
In this step we show  that it is sufficient to establish
 that for every $\zeta>0$
there exists $\varrho>0$ and $N_\zeta\in\N$ such that
\begin{equation}
\label{eq:suff1}
\bP \left(\frac{\bL(N)}{\tilde \bL (N)^{(q-1)/2}}
\sum_{\ui \in \{0, \ldots, N \}^q}   V_N(\ui) \gd_{\ui} 
\, \ge \, \varrho \bigg \vert \, N \in \tau \right)\, \ge \, 1-\zeta,
\end{equation}
for $N \ge N_\zeta$.

Notice in fact  that 
%%F
$\hat \bbE_\tau X= \mbeta^q \sum_{\ui} V_k(\ui) \gd_{\ui}$,
where $\ui \in \{d, \ldots, f\}^q$. %, but if $\tau$ is chosen according 
%to $\bP_{d,f}$, the summation is effectively over $\ui \in \{d, \ldots, f\}^q$.
Since $V_k(\ui)$ is invariant under the transformation $\ui=(i_1, \ldots, 
i_q) \mapsto (i_1+n, \ldots, 
i_q+n)$ (any $n \in \bbZ$), we may very well work on $\{0, \ldots, f-d\}$,
that is on an interval $\{0, \ldots, N\}$ ($\gep k \le N \le k$) and $\tau$ is a renewal
with $\tau_0=0$ and conditioned to $N \in \tau$. With this change of
variables, \eqref{eq:fromCE} reads
\begin{equation}
\label{eq:forCE1}
\bP \left(
\mbeta^q \sum_{\ui \in \{0, \ldots, N\}^q}
 V_k(\ui) \gd_{\ui} \, \ge \, a A^{(q-1)/2} \, \bigg \vert N \in \tau
\right) \,
\ge \, 1- \zeta. 
\end{equation}
Now two observations are in order:
%%F
\begin{itemize}
\item $V_k(\ui)/V_N(\ui)= (N/k)^{1/2} [\tilde \bL (N)/ \tilde \bL(k)]^{(q-1)/2}$ so that for
$k$ sufficiently large (that is for $A$ larger than a constant depending on
$\gep $ and $L(\cdot)$) we have
\begin{equation}
\frac{V_k(\ui)}{V_N(\ui)} \, \ge \, \frac{\gep^{1/2}}2.
\end{equation} 
\item By \eqref{eq:2beused},
 \eqref{eq:k} and \eqref{eq:Lb} we see that 
 \begin{equation}
\mbeta^q \, \ge \, 2^{-q} A^{(q-1)/2} \, \frac{L(k-1)}{\tilde L(k-1)^{(q-1)/2}}
\, \ge \,  2^{-q} \const _L A^{(q-1)/2} \, \frac{\bL(N)}{\tilde \bL(N)^{(q-1)/2}}
\, .
\end{equation}
\end{itemize}

These two observations show that for $A$ sufficiently large 
\eqref{eq:forCE1} is implied by
\begin{equation}
\label{eq:for CE2}
\bP \left(
 \frac{\bL(N)}{\tilde \bL(N)^{(q-1)/2}} 
 \sum_{\ui \in \{0, \ldots, N\}^q}
 V_N(\ui) \gd_{\ui} \, \ge\, \frac{2^{q}}{\const_L \gep^{1/2}} a
\bigg\vert \, N \in \tau 
\right) \, \ge \, 1-\zeta\, .
\end{equation}
Therefore, at least if 
$A$
is  larger than a suitable constant depending on  $\gep$ and $L(\cdot)$,
it is  sufficient to
prove \eqref{eq:suff1}.

\medskip

\noindent
{\it Step 2: removing the constraint.}
In this step we claim that there exists a positive constant $c$,
that depends only on $L(\cdot)$, such that
if 
\begin{equation}
\label{eq:suff2}
\bP \left(\frac{\bL(N)}{\tilde \bL (N)^{(q-1)/2}}
\sum_{\ui \in \{1, \ldots, \lfloor N/2\rfloor \}^q}   V_N(\ui) \gd_{\ui} 
\, \ge \, \varrho \, \right)\, \ge \, 1- c \zeta,
\end{equation}
then \eqref{eq:suff1} holds. Note first of all that the 
random variable that we are estimating is smaller (since 
$V_N(\cdot) \ge 0$)
than the random variable  in \eqref{eq:suff1}, for every given $\tau$-trajectory.
It is therefore sufficient to bound the 
Radon-Nykodym derivative of the law of $\tau\cap [0, \lfloor N/2
\rfloor]$ without constraint
$N \in \tau$ with respect to the law of the same random set with the constraint.
Such an estimate can be found for example in 
\cite[Lemma A.2]{cf:GLTmarg}. %. or \cite[...]{cf:DGLT}. 

\medskip

\noindent{\it Step 3: reduction to a convergence in law statement.}
For $\rho:= 1/(2(q-1))$ we define the subset $S_\rho (N)$ of 
$\sort(\{ 0,1, \ldots, N \}^q)$ (recall that the latter is the set of 
increasingly rearranged $\ui$ vectors) such that 
$i_j \le N ((j-1)\rho +(1/2))$ for $j=1, 2, \ldots, q$. 

The claim of this step is that \eqref{eq:suff2} follows if 
\begin{equation}
\label{eq:suff3}
\eta_N \, :=\, \frac{\bL (N)}{\tilde \bL (N)^{(q-1)/2}}
\sum_{\ui \in S_\rho (N)} V_N(\ui) \gd_{\ui} \stackrel{N \to \infty}\Longrightarrow 
\eta _\infty \ \ \ \text{ with } \ \eta_\infty\, >\, 0 \text{ a.s.}\, ,
\end{equation}
where $\Longrightarrow$ denotes convergence in law. 
 
In order to see why \eqref{eq:suff3} implies \eqref{eq:suff2} it
suffices to observe that replacing $N$ with $\lfloor N/2\rfloor$ in
\eqref{eq:suff2} (except when it already appears as $\lfloor
N/2\rfloor$) introduces an error that can be bounded by a
multiplicative constant (say, 2) for $N$ sufficiently large, so that
it suffices to show that $\bP( \eta_N \ge 2 \varrho) \ge 1-c \zeta$.
But \eqref{eq:suff3} yields $\lim_N\bP( \eta_N \ge 2 \varrho) \ge \bP(
\eta_\infty \ge 3 \varrho)$.  At this point if we choose
$\varrho:=\varrho(\zeta)$ such that $ \bP( \eta_\infty \ge 3 \varrho)=
1-(c\zeta/2)$, we are assured that for $N$ sufficiently large (how
large depends on $\zeta$) $\bP( \eta_N \ge 2 \varrho) \ge 1-c \zeta$
and we are reduced to proving \eqref{eq:suff3}.

\medskip

\noindent
{\it Step 4: proof of the convergence in law statement \eqref{eq:suff3}}.
This step depends on the following  lemma, that we prove just below:

\medskip

\begin{lemma}
\label{th:CE}
For every $\theta_0 \in (0,1)$ we have
\begin{equation}
  \lim_{N \to \infty} \sup_{\theta \in [\theta_0,1]}
  \bE\left[ \left (
      \frac 1{\tilde \bL (N)} \sum_{j=1}^{\lfloor \theta N\rfloor} 
R_{1/2}(j) {\gd_j} \, - \, \frac{\const}{2\pi}
      \right )^2 \right]\, =\, 0\, ,
\end{equation}
with $\const := \lim_{x \to \infty} \bL(x)/ L(x)(\in [1,\const _L^{-1} ])$.
\end{lemma}
\medskip

For $p=1,2, \ldots, q$ we introduce the random variables
\begin{multline}
\eta_{N, p}\, :=\\
\left(\frac{2\pi}{\const}\right)^{p-q} \frac{\bL(N)}{N^{1/2}\tilde \bL (N)^{p-1}}
\sum_{i_1=0}^{ [N/2 ]}\sum_{i_2 =i_1+1}^{[(\rho+(1/2))N] }\ldots 
\sum_{i_p=i_{p-1}+1}^{[((p-1)\rho+(1/2))N] } \gd_{i_1}
\prod_{r=2}^p R_{1/2}\left(i_r-i_{r-1}\right) \gd_{i_r}\, ,
\end{multline}
where the product in the right-hand side has to be read as $1$ if $p=1$ and, in this case, there is only the sum over $i_1$. 
First of all remark that
%%F
$\eta_{N, q}= \sqrt{q!}\eta_N$ (recall \eqref{eq:useR2})
 and that $\eta_{N, p-1}$ is obtained 
from $\eta_{N,p}$ by removing the last term in the product, the corresponding
sum and  by multiplying by  $ 2\pi \tilde \bL(N)/\const$.
We now claim that Lemma~\ref{th:CE} implies that for $p=2, 3, \ldots, q$
\begin{equation}
\label{eq:L1andind}
\lim_{N \to \infty} \bE \left[ \left \vert \eta_{N,p}- \eta_{N, p-1}\right\vert \right] \, =\, 0\, , 
\end{equation}
which clearly reduces the problem of proving
$\eta_N \Longrightarrow \eta_\infty$ to
proving $\eta_{N,1} \Longrightarrow \eta_\infty$, and  $\eta_\infty$ has to be 
a positive random variable.
But in fact we have
\begin{equation}
\label{eq:convlaw}
(2\pi /\const)^{q-1} \frac{L(N)}{\bL(N)}\, \eta_{N,1}\, =\, 
\frac{L(N)}{\sqrt{N}}
\sum_{i=0}^{\lfloor N/2\rfloor} \gd_i 
\stackrel{N \to \infty}{\Longrightarrow}
\frac{1}{2\sqrt{\pi}} \vert Z \vert  \ \ \ \ \ \ \ \ ( Z \sim \cN(0,1)).
\end{equation}
The convergence in \eqref{eq:convlaw} is a standard result that
we outline briefly. First of all for every choice of $n, m \in \N$ we
have
\begin{equation}
\label{eq:deltatau}
\left\{ \sum_{i=1}^n \gd_i \, < \, m \right\} \, =\, \left\{ \tau_m > n \right\} ,
\end{equation} 
so that the asymptotic law of the {\sl normalized local time}
of $\tau$ up to $n$, {\sl i.e.} $L(n)n^{-1/2}\sum_{i=1}^n \gd_i$, is directly linked
to the domain of attraction of the random variable $\tau_1$. 
Explicitly, one directly verifies that for $\gl>0$
\begin{equation}
\bE \left[ \left( 1- \exp(-\gl \tau_1)\right)\right] \stackrel{\gl \searrow 0} \sim
2\sqrt{\pi}L(1/\gl) \sqrt{\gl},
\end{equation}
so that, if $a(\cdot)$ is the asymptotic inverse of the regularly varying function $r(\cdot)$, defined by
$r(x):=\sqrt{x}/L(x)$ for  $x>0$, that is $a(r(x))\sim r(a(x))\sim x$
for $x \to \infty$, we have 
\begin{equation}
\lim_{N \to \infty}
\bE \left[ \exp \left( -\gl \tau_N /a(N) \right) \right] \, =\, \exp\left(-2 \sqrt{\pi \gl}\right) \, =\, 
\bE \left[ \exp(-\gl Y)\right],
\end{equation}
where $Y$ is a positive random variable with density $f_Y(y)$ equal to
$y^{-3/2} \exp(-\pi/y)$ (for $y>0$).  On the other hand for $t >0$
by \eqref{eq:deltatau} we have
\begin{equation}
\bP\left( \frac{L(n)}{\sqrt{n}}
\sum_{j=1}^n \gd_j <t\right) \stackrel{n \to \infty}\sim
\bP\left( \tau_{\lfloor t \sqrt{n}/L(n)\rfloor} >n \right).
\end{equation}
Therefore
if we observe that $a(t \sqrt{n}/L(n)) \sim t^2 a( \sqrt{n}/L(n))
\sim t^2 n$, for $n \to \infty$, we directly obtain that
\begin{equation}
\lim_{n \to \infty}
\bP\left( \frac{L(n)}{\sqrt{n}}
\sum_{j=1}^n \gd_j <t\right) \stackrel{n \to \infty}\sim
\bP\left(Y\, >\, \frac 1{t^2} \right).
\end{equation}
By using the (explicit) density  of $Y$, 
one directly verifies that $\bP (Y >  1/t^2 )$ coincides with
$\bP ( \vert Z\vert/\sqrt{2\pi} < t)$ for every $t>0$, that is  
 \eqref{eq:convlaw} is established (recall that in \eqref{eq:convlaw}
 the summation is up to $N/2$).
\medskip

We are therefore left with
proving   \eqref{eq:L1andind}. This  follows by
observing that for $p =3,4, \ldots,q$
\begin{multline}
\label{eq:forL1andind}
\bE \left[ \left \vert \eta_{N,p}- \eta_{N, p-1}\right\vert \right] \, \le \,
\left(\frac{2\pi}{\const}\right)^{p-q} \frac{\bL(N)}{N^{1/2}\tilde \bL
  (N)^{p-2}} \times \\
   \sum_{i_1=0}^{ \lfloor N/2 \rfloor}\sum_{i_2
  =i_1+1}^{\lfloor(\rho+(1/2))N\rfloor }\ldots
\sum_{i_{p-1}=i_{p-2}+1}^{\lfloor ((p-2)\rho+(1/2))N\rfloor } \bE \left[ \gd_{i_1}
  \prod_{r=2}^{p-1} R_{1/2}\left(i_r-i_{r-1}\right) \gd_{i_r}\right]\,
\times
\\
\bE \left[ \left\vert \frac 1{ \tilde\bL (N)}
    \sum_{i_{p}=i_{p-1}+1}^{\lfloor((p-1)\rho+(1/2))N\rfloor }
    R_{1/2}\left(i_p-i_{p-1}\right) \gd_{i_p} \, - \frac
    {\const}{2\pi} \right \vert \, \bigg \vert
  \gd_{i_{p-1}}=1\right]\, ,
\end{multline}
and the same expression holds if $p=2$ but in this case
the external summation is only over $i_1$
and $\prod_{r=2}^{p-1} R_{1/2}\left(i_r-i_{r-1}\right) \gd_{i_r}$ is replaced by $1$.
The bound \eqref{eq:forL1andind} follows from the triangular inequality
and from the renewal property of $\tau$.
Next, note that
\begin{multline}
\label{eq:forL1andind2}
\bE \left[ \left\vert \frac 1{ \tilde \bL (N)}
    \sum_{i_{p}=i_{p-1}+1}^{\lfloor((p-1)\rho+(1/2))N\rfloor }
    R_{1/2}\left(i_p-i_{p-1}\right) \gd_{i_p} \, - \frac
    {\const}{2\pi} \right \vert \, \bigg \vert
  \gd_{i_{p-1}}=1\right]\, =
\\
\bE \left[ \left\vert \frac 1{ \tilde \bL (N)}
    \sum_{i=1}^{\lfloor((p-1)\rho+(1/2))N\rfloor -i_{p-1}} R_{1/2}\left(i\right)
    \gd_{i} \, - \frac {\const}{2\pi} \right \vert \right] \stackrel{N
  \to \infty} \longrightarrow 0\, ,
\end{multline}
uniformly in the choice of $i_{p-1}\in \{ i_{p-2}+1, \ldots, \lfloor
((p-1)\rho+(1/2)N\rfloor \}$.  This is because the summation in
\eqref{eq:forL1andind2} contains at least $[\rho N]$ terms (and no
more than $N$) so that we can apply Lemma~\ref{th:CE}.  The fact that
$\bE \left[ \left \vert \eta_{N,p}- \eta_{N, p-1}\right\vert
\right]=o(1)$ as $N \to \infty$ is therefore a consequence of the
following explicit estimate:
\begin{multline}
  \frac{\bL(N)}{N^{1/2}\tilde \bL (N)^{p-2}} \sum_{i_1=0}^{ \lfloor
    N/2 \rfloor}\sum_{i_2 =i_1+1}^{\lfloor(\rho+(1/2))N\rfloor }\ldots
  \sum_{i_{p-1}=i_{p-2}+1}^{\lfloor ((p-2)\rho+(1/2))N\rfloor } \bE
  \left[ \gd_{i_1} \prod_{r=2}^{p-1} R_{1/2}\left(i_r-i_{r-1}\right)
    \gd_{i_r}\right]\, \le
  \\
  \frac{\bL(N) \const _L ^{-(p-1)}}{N^{1/2}\tilde \bL (N)^{p-2}}
  \sum_{i_1=0}^{ \lfloor N/2 \rfloor}\sum_{i_2
    =i_1+1}^{\lfloor(\rho+(1/2))N\rfloor }\ldots
  \sum_{i_{p-1}=i_{p-2}+1}^{\lfloor((p-2)\rho+(1/2))N\rfloor } R_{1/2}(i_1)
  \prod_{r=2}^{p-1} \left(R_{1/2}\left(i_r-i_{r-1}\right)\right)^2 \\
%%F 
 \stackrel{N \to \infty}\sim \sqrt2 \const _L ^{-(p-1)}\, ,
\end{multline}
where we have used the definition \eqref{eq:Ltilde} of the slowly varying
function $\tilde \bL(\cdot)$ and the fact that $\int_0^x (y^{1/2}\bL(y))^{-1}
\dd y \stackrel{x \to \infty}\sim 2 x^{1/2}/\bL (x)$.
This completes the proof of Lemma~\ref{th:fromCE}.
\qed

\medskip

\noindent
{\it Proof of Lemma~\ref{th:CE}.}
This is very similar to the proof of Lemma~5.4 in \cite{cf:GLTmarg}
(that, in turn generalizes a result of K.~L.~Chung and P.~Erd\"os
\cite{cf:chungerdos}). We give it in detail in order to clarify the role of the
slowly varying function.

First of all let us remark that 
\begin{equation}
\frac 1{\tilde \bL (N)} \sum_{j=1}^{[\theta N]} R_{1/2}(j) \bE\left[\gd_j\right] \stackrel{N\to \infty}\sim
 \frac{\const \tilde \bL (\theta N)}{2\pi \tilde \bL (N)}  \stackrel{N\to \infty}\sim 
 \frac{\const}{2\pi }\, , 
\end{equation}
where the last asymptotic relation holds uniformly in $\theta$, when $\theta$ lies  in a compact subinterval of $(0, \infty)$. 
The statement is therefore reduced to showing that
the variance of
\begin{equation}
Y_n\, :=\, \sum_{j=1}^{n} R_{1/2}(j) {\gd_j},
\end{equation}
is $o(\tilde L(n)^2)$. 

Let us compute and start by observing that
\begin{equation}
\begin{split}
\text{var}_\bP \left( Y_n \right) \, &=\, \sum_{i,j=1}^n R_{1/2}(i)R_{1/2}(j)\left[
\bE\left[ \gd_i \gd_j\right] - \bE\left[ \gd_i \right]   \bE\left[\gd_j\right]
\right]
\\ 
&=\, 
2\sum_{i=1}^{n-1} \sum_{j=i+1}^n 
R_{1/2}(i)R_{1/2}(j)\left[
\bE\left[ \gd_i \gd_j\right] - \bE\left[ \gd_i \right]   \bE\left[\gd_j\right]
\right]\, 
%%F
+\, O(\tilde L(n)) \, \\
&=: \, 2 T_n + O(\tilde L(n)),
\end{split}
\end{equation}
and 
\begin{equation}
\begin{split}
T_n \, &=\, \sum_{i=1}^{n-1} R_{1/2}(i){ \bE\left[ \gd_i \right] }
\left[
\sum_{j=1}^{n-i} R_{1/2}(i+j){ \bE\left[ \gd_j \right] } -
\sum_{j=i+1}^{n} R_{1/2}(j){ \bE\left[ \gd_j \right] } 
\right]
\\
& \le \, 
 \sum_{i=1}^{n-1} R_{1/2}(i){ \bE\left[ \gd_i \right] }
\left[
\sum_{j=1}^{n-i} R_{1/2}(i+j){ \bE\left[ \gd_j \right] } -
\sum_{j=i+1}^{n} R_{1/2}(i+j){ \bE\left[ \gd_j \right] } 
\right]
\\
& \le \, 
 \sum_{i=1}^{n-1} R_{1/2}(i){ \bE\left[ \gd_i \right] }
\sum_{j=1}^{i} R_{1/2}(i+j){ \bE\left[ \gd_j \right] } 
\, \le \, 
 \sum_{i=1}^{n-1} \left(R_{1/2}(i)\right)^2{ \bE\left[ \gd_i \right] }
\sum_{j=1}^{i} { \bE\left[ \gd_j \right] } 
\\ 
& \phantom{movemovemove} \le \, \const _L ^{-2}
 \sum_{i=1}^{n-1} \left(R_{1/2}(i)\right)^3
 \sum_{j=1}^{i} R_{1/2}(j)\stackrel{n \to \infty}\sim
 2\const _L ^{-2} \int_{0}^n \frac1{(1+x) (\bL (x))^4} \dd x\, , 
\end{split}
\end{equation}
where the first three inequalities follow since $R_{1/2}(\cdot)$ is non increasing and the fourth follows from \eqref{eq:Doney-bound}.
The conclusion of the proof follows now from Remark~\ref{rem:auxL}.
\qed

\medskip

\begin{rem}
\label{rem:auxL}
\rm
For $x \to \infty$
\begin{equation}
\int_0^x \frac1{(1+y) (\bL(y))^4} \dd y \, \ll \, \left( \tilde \bL (x) \right)^2,
\end{equation}
with $ \tilde \bL (x)$ defined as in \eqref{eq:Ltilde} with $L(\cdot)$
replaced by $\bL(\cdot)$.
This is a consequence of \eqref{eq:Ltilde-prop} (which of course holds also for $\bL(\cdot)$):
\begin{equation}
\int_0^x \frac1{(1+y) (\bL(y))^4} \dd y \, \ll \, 
\int_0^x \frac1{(1+y) (\bL(y))^2} \tilde \bL (y) \dd y\, \le \, \tilde \bL (x)
\int_0^x \frac1{(1+y) (\bL(y))^2} \dd y \, ,
\end{equation}
and the rightmost term is $\left(  \tilde \bL (x)\right)^2$.
\end{rem}

\section{A general monotonicity result}
\label{sec:monotonicity}

We present now a very general result: we give it in our context but a look at the proof suffices
to see 
that it   holds also under substantially  milder assumptions 
 on the process $\tau$.
\medskip

\begin{proposition}\label{th:monotonicity}
The free energy $\tf(\gb,h)$ is a non-increasing function of $\gb$ on $[0,\infty)$.
Therefore
\begin{itemize}
 \item[(i)] $\gb\mapsto h_c(\gb)$ is a non-decreasing function of $\gb$.
 \item[(ii)] There exists a critical value $\gb_c\in[0,\infty]$ such that $h_c(0)=h_c(\gb)$ if and only if $\gb\le \gb_c$.
\end{itemize}
\end{proposition}
\medskip

This result is of particular relevance when $\sum_n 1/(n L(n)^2) <\infty$,
that is when for small $\gb$ we have $h_c(\gb)= h_c(0)$
({\sl cf.} \S~\ref{sec:review-rs}): in this case $\gb_c$
is the transition point from the  irrelevant disorder regime to the  relevant one. 
But also in our set-up, in which 
$\sum_n 1/(n L(n)^2) =\infty$, it is  of some use since it implies 
that it is sufficient to prove Theorem~\ref{th:main} for one value of $\gb_0>0$
and the statement holds also for any other value of $\gb_0$
(by accepting, of course,   a worse estimate on the shift of the critical point if one follows the estimates quantitatively, see Remark~\ref{rem:lazy}). 

\begin{proof}
We just need to prove that $\gb\mapsto\tf(\gb,h)$ is a non-increasing function on $[0,\infty)$ as the other points are trivial consequence of this result.
To do so, we prove that $\gb\mapsto \bbE[\log Z_{N,\go}]$ is a non-increasing function of $\gb$, and pass to the limit. 
The proof is the adaptation of an argument  used  in \cite{CY} for directed polymers with bulk disorder  to prove a similar result.

What we will show is 
\begin{equation}
\label{eq:fCY}
 \frac{\partial}{\partial \gb} \bbE\left[  \log Z_{N,\go} \right] \,=\, \bbE\left[ \frac{\partial}{\partial \gb} \log Z_{N,\go} \right]\,\le\, 0.
\end{equation}
The proof of the equality
in \eqref{eq:fCY} is standard and can be easily adapted from \cite[Lemma 3.3]{CY}.
Recall now that $\mbeta:=M'(\gb)/M(\gb)$. We have
\begin{equation}\begin{split}
\bbE\left[ \frac{\partial}{\partial \gb}\log Z_{N,\go}\right]&=\bE \left[\bbE\left[ \frac{1}{Z_{N,\go}}\sum_{n=1}^N(\go_n-\mbeta)\gd_n\exp\left(\sum_{n=1}^N  [\gb\go_n+h-\log M(\gb)]\gd_n\right)\gd_N\right]\right]\\
&=\bE\left[\exp\left(\sum_{n=1}^N h \gd_n\right)\gd_N \hat\bbE_{\tau}\left[Z_{N,\go}^{-1}\sum_{n=1}^N(\go_n-\mbeta)\gd_n\right]\right].
\end{split}\end{equation}
For a fixed trajectory of the renewal, the  probability measure $\hat \bbP_{\tau}$ (recall definition \eqref{eq:ptau}), is a product measure, so that, since
%\cite{cf:Preston}: 
 $Z_{N,\go}^{-1}$ is a decreasing function of $\go$ and 
 $\sum_{n=1}^N(\go_n-\mbeta)\gd_n$ is 
a non-decreasing  function of $\go$, by the Harris--FKG inequality we have
\begin{equation}
\hat\bbE_{\tau}\left[Z_{N,\go}^{-1}\sum_{n=1}^N(\go_n-\mbeta)\gd_n\right]\le \hat\bbE_{\tau}\left[Z_{N,\go}^{-1}\right]\hat\bbE_{\tau}\left[\sum_{n=1}^N(\go_n-\mbeta)\gd_n\right]=0.
\end{equation}

\end{proof}

\section*{Acknowledgments}
This work has been  supported by ANR, grant {\sl POLINTBIO}.
F. L. T. is supported also by the ANR grant LHMSHE.

\end{document}